\documentclass{IEEEtran}
\usepackage{cite}
\usepackage{amsmath,amssymb,amsfonts}
\usepackage{algorithmic}
\usepackage{textcomp}
\def\BibTeX{{\rm B\kern-.05em{\sc i\kern-.025em b}\kern-.08em
    T\kern-.1667em\lower.7ex\hbox{E}\kern-.125emX}}

\usepackage[utf8]{inputenc}

\usepackage{amsthm}
\usepackage{comment}

\newtheorem{theorem}{Theorem}[section]
\newtheorem{lemma}[theorem]{Lemma}
\newtheorem{proposition}[theorem]{Proposition}

\theoremstyle{definition}
\newtheorem{definition}{Definition}[section]
\newtheorem{example}{Example}[section]
\newtheorem{remark}{Remark}
\newtheorem{problem}{Problem}

\usepackage{graphicx}      
\usepackage[dvipsnames]{xcolor}
\usepackage{paralist}
\usepackage{csquotes}
 \usepackage{array}

\usepackage{xspace}
\usepackage{xifthen}
\usepackage{url}

\usepackage{pgf}
\usepackage{tikz}
\usetikzlibrary{trees,decorations,arrows,automata,shadows,positioning,plotmarks,backgrounds,shapes}
\usetikzlibrary{calc,matrix,fit,petri,decorations.markings,decorations.pathmorphing,patterns,intersections,decorations.text}

\tikzstyle{mystate}=[state,inner sep=1pt,minimum size=10pt,line width=0.2mm]
\tikzstyle{mysquare}=[inner sep=1pt,minimum size=10pt,line width=0.2mm]
\tikzstyle{mybstate}=[rectangle,rounded corners,minimum size=15pt,line width=0.2mm]
\tikzstyle{myestate}=[ellipse,inner sep=1pt,line width=0.2mm]
\newcommand{\SFSAutomatEdge}[5]{\path[->](#1) edge[#4,line width=0.3mm] node[#5] {\ensuremath{#2}} (#3);}

\newcommand{\SFSAutomatEdgeCrossed}[5]{
\path[->](#1) edge[#4,line width=0.3mm] node[#5] {\ensuremath{#2}}  coordinate[pos=0.5] (A) (#3);
\draw[-] (A)+(+0.01,0.06) edge[line width=0.3mm] (A);
\draw[-] (A)+(-0.01,-0.06) edge[line width=0.3mm] (A);
}

\definecolor{sws-blue}{RGB}{83, 151, 204} %
\definecolor{sws-red}{RGB}{197,14,31} 
\definecolor{sws-green}{RGB}{82,171,54}
\newcommand{\lgreen}{sws-green!20}
\newcommand{\gr}{black!20}

\newcommand{\REFlem}[1]{\text{Lem.~\ref{#1}}}
\newcommand{\REFthm}[1]{\text{Thm.~\ref{#1}}}
\newcommand{\REFdef}[1]{Def.~\ref{#1}}

\newcommand{\REFsec}[1]{Sec.~\ref{#1}}
\newcommand{\REFprop}[1]{Prop.~\ref{#1}}
\newcommand{\REFfig}[1]{Fig.~\ref{#1}}

\newcommand{\REFproblem}[1]{Problem~\ref{#1}}
\newcommand{\REFex}[1]{Example~\ref{#1}}

\newcommand{\D}{\mathop{:=}}

\newcommand{\dotcup}{\dot{\cup}}

\newcommand{\ASUBS}[1]{\raisebox{0.1ex}[0.5ex][-0.2ex]{$\scriptstyle\mathrm{#1}$}}
\newcommand{\ALPHABET}[1]{\Sigma^{\raisebox{0.2ex}[0.0ex][0.0ex]{\phantom{$\scriptstyle *$}}}_{\ASUBS{#1}}}
\newcommand{\ASigma}{\Sigma}
\newcommand{\ASigmac}{\ALPHABET{c}}
\newcommand{\ASigmauc}{\ALPHABET{uc}}

\newcommand{\LPrefix}[1]{\mathop{\mathrm{pfx}}{(#1)}}

\newcommand{\SGenerated}[1]{\ON{\mathbf{s}}\ifthenelse{\isempty{#1}}{}{(#1)}}

\newcommand{\GG}{G}

\newcommand{\Aut}{M}

{\end{list}}

{\end{list}}

{\end{list}}

\newcommand\set[1]{\{ #1 \}}
\newcommand\tuple[1]{{\langle #1 \rangle}}

\let\exampleOrig\endexample
\def\endexample{\hspace*{0pt}\hfill$\triangleleft$\exampleOrig}

\newcommand{\ON}[1]{\mathsf{#1}}

\def\clap#1{\hbox to 0pt{\hss#1\hss}}

\makeatletter 
\newif\ifFIRST
\newif\ifSECOND
\let\LISTOP\relax
\newcommand{\List}[4][\;]{#3#1%
        \FIRSTtrue
        \@for\i:=#2\do{%
        \ifFIRST\LISTOP{\i}\FIRSTfalse\else,\LISTOP{\i}\fi%
        }%
        #1#4%
        \let\LISTOP\relax
}
\makeatother
\newcounter{DINGLIST}
\stepcounter{DINGLIST}
\newcommand{\markD}[3][\;\;]{\text{\ding{\the\numexpr171+\theDINGLIST}\stepcounter{DINGLIST}}#1#3}

\makeatletter

\newcommand{\propNeg}{\@ifstar\propNegStar\propNegNoStar}
\newcommand{\propNegStar}[1]{\ensuremath{\left(\propNegNoStar{#1}\right)}}
\newcommand{\propNegNoStar}[2][\cdot]{\ensuremath{\neg\ifthenelse{\isempty{#2}}{#1}{#2}}}

\newcommand{\propConj}{\@ifstar\propConjStar\propConjNoStar}
\newcommand{\propConjStar}[2]{\ensuremath{\left(\propConjNoStar{#1}{#2}\right)}}
\newcommand{\propConjNoStar}[3][\cdot]{\ensuremath{\ifthenelse{\isempty{#2}}{#1}{#2}\wedge\ifthenelse{\isempty{#3}}{#1}{#3}}}

\newcommand{\propDisj}{\@ifstar\propDisjStar\propDisjNoStar}
\newcommand{\propDisjStar}[2]{\ensuremath{\left(\propDisjNoStar{#1}{#2}\right)}}
\newcommand{\propDisjNoStar}[3][\cdot]{\ensuremath{\ifthenelse{\isempty{#2}}{#1}{#2}\vee\ifthenelse{\isempty{#3}}{#1}{#3}}}

\newcommand{\propImp}{\@ifstar\propImpStar\propImpNoStar}
\newcommand{\propImpStar}[2]{\ensuremath{\left(\propImpNoStar{#1}{#2}\right)}}
\newcommand{\propImpNoStar}[3][\cdot]{\ensuremath{\ifthenelse{\isempty{#2}}{#1}{#2}\Rightarrow\ifthenelse{\isempty{#3}}{#1}{#3}}}

\newcommand{\propAequ}{\@ifstar\propAequStar\propAequNoStar}
\newcommand{\propAequStar}[2]{\ensuremath{\left(\propAequNoStar{#1}{#2}\right)}}
\newcommand{\propAequNoStar}[3][\cdot]{\ensuremath{\ifthenelse{\isempty{#2}}{#1}{#2}\Leftrightarrow\ifthenelse{\isempty{#3}}{#1}{#3}}}

\newcommand{\AllQ}{\@ifstar\AllQStar\AllQNoStar}
\newcommand{\AllQStar}[3][\;]{\ensuremath{\left(\forall #2#1.#1#3\right)}}
\newcommand{\AllQNoStar}[3][\;]{\ensuremath{\forall #2#1.#1#3}}
\newcommand{\AllQu}{\@ifstar\AllQuStar\AllQuNoStar}
\newcommand{\AllQuStar}[3][\;]{\ensuremath{\left(\forall^{\infty} #2#1.#1#3\right)}}
\newcommand{\AllQuNoStar}[3][\;]{\ensuremath{\forall^{\infty} #2#1.#1#3}}

\newcommand{\ExQ}{\@ifstar\ExQStar\ExQNoStar}
\newcommand{\ExQStar}[3][\;]{\ensuremath{\left(\exists #2#1.#1#3\right)}}
\newcommand{\ExQNoStar}[3][\;]{\ensuremath{\exists #2#1.#1#3}}

\newcommand{\NExQ}{\@ifstar\NExQStar\NExQNoStar}
\newcommand{\NExQStar}[3][\;]{\ensuremath{\left(\nexists #2#1.#1#3\right)}}
\newcommand{\NExQNoStar}[3][\;]{\ensuremath{\nexists #2#1.#1#3}}

\newcommand{\UniqueQ}{\@ifstar\UniqueQStar\UniqueQNoStar}
\newcommand{\UniqueQStar}[3][\;]{\ensuremath{\left(\exists! #2#1.#1#3\right)}}
\newcommand{\UniqueQNoStar}[3][\;]{\ensuremath{\exists! #2#1.#1#3}}

  \newlength{\SFS@HEIGHT}
  \newlength{\SFS@WIDTH}
  \newcommand{\SplitX}[2]{
            \settoheight{\SFS@HEIGHT}{$#2$}
            \settowidth{\SFS@WIDTH}{$#2$}
            \mbox{\begin{tikzpicture}[baseline=(current bounding box.center)]
            \node[] (E) at (0,0) {$#1$};
            \node[inner sep=0pt] (F) at ($(E.south west)+(1ex,-1ex)+(3ex+.5\SFS@WIDTH,-\SFS@HEIGHT)$) {$#2$};
            \node[] (E) at (0,0) {\phantom{$#1$}};
            \draw[fill] ($(E.east)+(1ex,0ex)$) circle (.2ex);
            \draw[-] ($(E.east)+(1ex,0ex)$) -- ($(E.south east)+(1ex,-0.5ex)$) -- ($(E.south west)+(1ex,-0.5ex)$) -- ($(E.south west)+(1ex,-1ex)-(0,\SFS@HEIGHT)$) -- ($(E.south west)+(2.5ex,-1ex)-(0,\SFS@HEIGHT)$);
            \draw[fill] ($(E.south west)+(2.5ex,-1ex)-(0,\SFS@HEIGHT)$) circle (.2ex);
            \end{tikzpicture}}}
  \newcommand{\SplitS}[2]{
            \settoheight{\SFS@HEIGHT}{$#2$}
            \settowidth{\SFS@WIDTH}{$#2$}
            \mbox{\begin{tikzpicture}[baseline=(current bounding box.center)]
            \node[] (E) at (0,0) {$#1$};
            \node[inner sep=0pt] (F) at ($(E.south west)+(1ex,0.5ex)+(3ex+.5\SFS@WIDTH,-\SFS@HEIGHT)$) {$#2$};
            \end{tikzpicture}}}     

\newcommand{\UNION}{\@ifstar\UNIONStar\UNIONNoStar}
\newcommand{\UNIONStar}[2]{\ensuremath{\left(\UNIONNoStar{#1}{#2}\right)}}
\newcommand{\UNIONNoStar}[2]{\ensuremath{\ifthenelse{\isempty{#1}}{\cdot}{#1}\cup\ifthenelse{\isempty{#2}}{\cdot}{#2}}}

\newcommand{\UNIOND}{\@ifstar\UNIONDStar\UNIONDNoStar}
\newcommand{\UNIONDStar}[2]{\ensuremath{\left(\UNIONDNoStar{#1}{#2}\right)}}
\newcommand{\UNIONDNoStar}[2]{\ensuremath{\ifthenelse{\isempty{#1}}{\cdot}{#1}\uplus\ifthenelse{\isempty{#2}}{\cdot}{#2}}}

\newcommand{\SETMINUS}{\@ifstar\SETMINUSStar\SETMINUSNoStar}
\newcommand{\SETMINUSStar}[2]{\ensuremath{\left(\SETMINUSNoStar{#1}{#2}\right)}}
\newcommand{\SETMINUSNoStar}[2]{\ensuremath{\ifthenelse{\isempty{#1}}{\cdot}{#1}\setminus\ifthenelse{\isempty{#2}}{\cdot}{#2}}}

\newcommand{\INTERSECT}{\@ifstar\INTERSECTStar\INTERSECTNoStar}
\newcommand{\INTERSECTStar}[2]{\ensuremath{\left(\INTERSECTNoStar{#1}{#2}\right)}}
\newcommand{\INTERSECTNoStar}[2]{\ensuremath{\ifthenelse{\isempty{#1}}{\cdot}{#1}\cap\ifthenelse{\isempty{#2}}{\cdot}{#2}}}

\newcommand{\CARTPROD}{\@ifstar\CARTPRODStar\CARTPRODNoStar}
\newcommand{\CARTPRODStar}[2]{\ensuremath{\left(\CARTPRODNoStar{#1}{#2}\right)}}
\newcommand{\CARTPRODNoStar}[2]{\ensuremath{\ifthenelse{\isempty{#1}}{\cdot}{#1}\times\ifthenelse{\isempty{#2}}{\cdot}{#2}}}

\newcommand{\FINCOUNT}{\@ifstar\FinCountStar\FinCountNoStar}
\newcommand{\FinCountStar}[1]{\ensuremath{\#(\ifthenelse{\isempty{#1}}{\cdot}{#1})}}
\newcommand{\FinCountNoStar}[1]{\ensuremath{\#\left(\ifthenelse{\isempty{#1}}{\cdot}{#1}\right)}}

\makeatother 

\newcommand{\parfun}{\ensuremath{\ON{\rightharpoonup}}}
\newcommand{\fun}{\ensuremath{\ON{\rightarrow}}}
\newcommand{\SetComp}[3][]{\{#1#2#1\mid#1#3#1\}}

\newcommand{\twoup}[1]{\ensuremath{2^{#1}}}

  \newcommand{\fz}[1]{\ensuremath{h\ifthenelse{\isempty{#1}}{}{(#1)}}} 
  \newcommand{\fo}[1]{\ensuremath{g\ifthenelse{\isempty{#1}}{}{(#1)}}}  
  
    \newcommand{\q}{\ensuremath{x}}
  \newcommand{\Q}{\ensuremath{X}}
  \newcommand{\Qo}{\ensuremath{Q^1}}
  \newcommand{\Qz}{\ensuremath{Q^0}}
    \newcommand{\Fc}{\ensuremath{\mathcal{F}}}
    \newcommand{\FcB}{\ensuremath{\mathcal{F}^\mathtt{B}}}
    \newcommand{\FcR}{\ensuremath{\mathcal{F}^\mathtt{R}}}
    \newcommand{\FcS}{\ensuremath{\mathcal{F}^\mathtt{S}}}
\newcommand{\FcStrong}{\ensuremath{\mathcal{S}}}
\newcommand{\FcWeak}{\ensuremath{\mathcal{W}}}

  \newcommand{\OpPre}[1]{\mathop{\mathrm{pfx}(#1)}}
  \newcommand{\OpLim}[1]{\mathop{\mathrm{lim}(#1)}}
\newcommand{\OpCl}[1]{\mathop{\mathrm{clo}(#1)}}
\newcommand{\OpInf}[1]{\mathop{\mathrm{Inf}(#1)}}
  
    \newcommand{\Trz}{\ensuremath{\delta^0}}
  \newcommand{\Tro}{\ensuremath{\delta^1}}

  \newcommand{\Tuple}[2][]{\List[#1]{#2}{(}{)}}
  \newcommand{\Set}[2][]{\List[#1]{#2}{\{}{\}}}
   \newcommand{\Lomega}{\ensuremath{\mathcal{L}}}
   
   \newcommand{\Pomega}{\ensuremath{\mathcal{P}}}
 \newcommand{\Lstar}{\ensuremath{L}}  
\newcommand{\Pstar}{\ensuremath{P}}

\begin{document}
\title{Supervisory Controller Synthesis for Non-terminating Processes is an Obliging Game}

\author{Rupak Majumdar and Anne-Kathrin Schmuck
\thanks{R. Majumdar and A.-K. Schmuck are with MPI-SWS, Kaiserslautern, Germany (e-mail: \{rupak,akschmuck\}@mpi-sws.org). Authors are ordered alphabetically.}}

\maketitle

\begin{abstract}
We present a new algorithm to solve the supervisory control problem over non-terminating processes modeled
as $\omega$-regular automata.
A solution to this problem was obtained by Thistle in 1995 
which uses complex manipulations of automata. 
We show a new solution to the problem through a reduction to 
\emph{obliging games}, 
which, in turn, can be reduced to $\omega$-regular reactive synthesis. 
Therefore, our reduction results in a symbolic algorithm based on manipulating sets of states using tools from reactive synthesis.
\end{abstract}

\section{Introduction}\label{sec:intro}

\emph{Supervisory control theory} (SCT) is a branch of control theory which is concerned with the control
of discrete-event dynamical systems (DES) with respect to temporal specifications.
Given such a DES, SCT asks to synthesize a \emph{supervisor} that restricts the possible sequences of events
such that any remaining sequence fulfills a given specification.
The field of SCT was established by the seminal work of Ramadge and Wonham \cite{RW87} concerning the control of \emph{terminating processes}, 
i.e., systems whose behavior can be modeled by regular languages over finite words.
This setting is well understood and summarized in standard text books \cite{Cassandras,WonhamBook}.

Already 30 years ago, Thistle and Wonham extended the scope of SCT to \emph{non-terminating processes} \cite{TW1991}, 
i.e., to the supervision of systems whose behavior can be modeled by regular languages over \emph{infinite} words.
Non-terminating processes naturally occur in models of infinitely executing reactive systems; 
$\omega$-words allow convenient modeling of both safety \emph{and} liveness specifications for such systems. 
In a sequence of papers \cite{TW1991,TW1994a,TW1994b} culminating in \cite{thistle1995},
Thistle and Wonham laid out the foundations for SCT over non-terminating processes and showed,
in particular, an algorithm to synthesize supervisors for general $\omega$-regular specifications under general $\omega$-regular plant properties.
This symbolic synthesis algorithm involves an intricate
fixed point computation over the $\omega$-regular languages, using structural operations on finite-state automata representations.
As a direct result of the complexity of the involved operations, 
to the best of our knowledge, this most general algorithm has never been implemented. 
%

A key observation in Thistle and Wonham's original work was the relationship between SCT and Church's problem from logic \cite{Church57}, 
and hence to techniques from reactive synthesis. This arguable also influenced early works on using logical formalisms to control hybrid systems \cite{maler1995synthesis,lygeros1999controllers,davoren2000logics} which partially culminated into the now well established field of formal methods for hybrid control.
However, despite this lasting connection between hybrid control and logic, the connection between the research fields of \emph{discrete} supervisory control and reactive synthesis got mostly lost over time and is just about to be re-established. This paper continues these recent efforts \cite{EhlersJdeds, DBLP:conf/ifac/SchmuckMS20, FabianAkesson2019comparative,ciolek2019compositional} by showing a new clean algorithmic connection between SCT and recent enhancements of reactive synthesis, called \emph{obliging games}.


While, conceptually, supervisor synthesis and reactive synthesis seem very similar, the resulting algorithmic reduction is not very obvious.
To understand the source of difficulty, let us recall the setting of the problem.
We are given a finite state machine that forms the transition structure of the DES for the synthesis problem, and we are given \emph{two} $\omega$-regular languages
defined over this machine.
The first language (let's call it $A$) models \emph{assumptions} on the plant: a supervisor can assume that the (uncontrolled) plant language will satisfy this assumption.
The second language (call it $S$) provides the specification that the supervisor must uphold whenever the plant operates in accordance to the assumptions, by preventing certain controllable events over time.

One can easily transform the given finite state machine to a two-player game, as in reactive synthesis, and naively ask for a winning strategy for
the winning condition $A \Rightarrow S$, which states that if the plant satisfies its assumption, then the resulting behavior satisfies
the specification.
While this reduction seems natural, it is incorrect in the context of SCT. 
The problem is that a control strategy may ``cheat'' and enforce the above implication vacuously by actively preventing the plant from satisfying the assumption.
In SCT, such undesired solutions are ruled out by a \emph{non-conflicting} requirement: any finite word compliant with the supervisor
must be extendable to an infinite word that satisfies $A$. Hence, a non-conflicting supervisor always allows the plant
to fulfill the assumption. 
The non-conflicting requirement is not a linear property \cite{EhlersJdeds}, and cannot be \enquote{compiled away} in reactive synthesis.

The main contribution of this paper is a reduction of the supervisory control problem to a class of reactive synthesis problems called \emph{obliging games} \cite{chatterjee2010obliging} that precisely
capture a notion of non-conflicting strategies in the context of reactive synthesis.
The main result of \cite{chatterjee2010obliging} shows that obliging games can be reduced to usual reactive synthesis on a larger game.
Once the intuitive connection between supervisory control and obliging games is made, the formal reduction is almost trivial. 

We consider this simplicity as a \emph{feature} of our work: our conceptual reduction from supervisory control to obliging games, and hence to reactive synthesis,
forms a separation of concerns between (a) the modeling of specifications and non-conflicting strategies and (b) the (non-trivial, but well-understood)
algorithmics of solving games.

\smallskip
\textit{Other Related Work:}
This paper continues recent efforts in establishing a formal connection between reactive synthesis and SCT for \emph{terminating processes} \cite{EhlersJdeds, DBLP:conf/ifac/SchmuckMS20, FabianAkesson2019comparative,ciolek2019compositional} and \emph{non-terminating processes} \cite{SchmuckJdeds20}. While  \cite{SchmuckJdeds20} focusses on a language-theoretic connection, this paper establishes a connection between synthesis algorithms over automata realizations.

Within the SCT community, non-terminating processes have gained more attention in recent years, see e.g., \cite{Aucher_2014,vanHulst2017,SakakibaraUshio_2018,yang2020refinements}.
However, in all these works, the plant itself does not posses non-trivial liveness properties, which allows to transform the resulting synthesis problem to the usual setting of reactive synthesis. 
Notable exceptions are, e.g., \cite{baier2012hierarchical,Moor_Report_omegaSCT}, where synthesisis is restricted to \emph{deterministic} Büchi automata models, capturing only a strict subclass of $\omega$-regular properties.

Symbolic algorithms for GR(1) specifications satisfying a non-conflicting requirement were presented in \cite{majumdar2019environmentally}.
Their algorithm has the advantage of a ``direct'' implementation using symbolic manipulation of sets of states.
We leave as future work whether a similar direct algorithm can be designed for general obliging games.

\section{Preliminaries}\label{sec:SCT:prelim}

\noindent\textit{Formal Languages.}
Given a finite alphabet $\Sigma$, we write $\Sigma^*$, $\Sigma^+$, and $\Sigma^\omega$ for the sets of finite words, non-empty finite words,
and infinite words over $\Sigma$, and write $\Sigma^\infty = \Sigma^* \cup \Sigma^\omega$.
We call the subsets $L\subseteq\Sigma^*$ and $\mathcal L\subseteq\Sigma^\omega$ a $*$-language and an $\omega$-language over $\Sigma$, respectively.

We write $w\le v$ (resp., $w<v$) if $w$ is a prefix of $v$ (resp., a
strict prefix of $v$). 
The set of all prefixes of a word $w\in\Sigma^\infty$ is a $*$-language denoted by
$\OpPre{w}\subseteq \Sigma^*$. 
For $L\subseteq\Sigma^*$, we have $L\subseteq \OpPre{L}$. 
A $*$-language $L$ is \emph{prefix-closed} if $L = \OpPre{L}$.
The \emph{limit} $\OpLim{L}$ of a $*$-language $L$ is the $\omega$-langauge which contains all words $\alpha\in\Sigma^\omega$ which have infinitely many prefixes in $L$. We further define $\OpCl{\mathcal{L}} :=\OpLim{\OpPre{\mathcal{L}}}$ as the
\emph{topological closure} of $\mathcal{L}\subseteq\Sigma^\omega$. An $\omega$-language $\mathcal L$ is \emph{topologically-closed} if $\mathcal L=\OpCl{\mathcal{L}}$.

\smallskip
\noindent\textit{Finite State Machines.}
A \emph{finite state machine} is a tuple
$\Aut=(\Q,\,\Sigma,\,\delta,\, \q_0)$,
with \emph{state set} $\Q$,
\emph{alphabet} $\Sigma$,
\emph{initial state} $\q_0\in \Q$, 
and the partial \emph{transition function} 
$\delta : \Q\times \Sigma\parfun 2^{\Q}$.
For $\q\in\Q$ and $\sigma\in\Sigma$, we write $\delta(\q,\sigma)!$ to signify that $\delta(\q,\sigma)$ is defined. 
We call $\Aut$ \emph{deterministic} if $\delta(\q,\sigma)!$ implies $|\delta(\q,\sigma)|=1$. 
We call $\Aut$ \emph{non-blocking} if for all $\q\in\Q$ there exists at least one $\sigma\in\Sigma$ s.t.\ $\delta(\q,\sigma)!$.

A \emph{path} of $\Aut$ is a finite or infinite sequence $\pi=\q_0 \q_1 \hdots$ 
s.t.\ for all $k\in\ON{Length}(\pi)-1$ there exists some $\sigma_k\in\Sigma$ s.t.\ $\q_{k+1}\in\delta(\q_k,\sigma_k)$. 
If $\pi$ is finite, we denote by $\ON{Last}(\pi)=\q_n$ its last element. 
We collect all finite and infinite paths of the finite state machine $\Aut$ in the sets $\Pstar(\Aut)\subseteq \q_0\Q^*$ 
and  $\Pomega(\Aut)\subseteq \q_0\Q^\omega$, respectively. 
Given a string 
$s=\sigma_0\sigma_1\hdots\in\Sigma^\infty$ we say that a path $\pi$ of $\Aut$ is \emph{compliant} with 
$s$ if $\ON{Length}(s)=\ON{Length}(\pi)-1$ and for all $k\in\ON{Length}(\pi)-1$ we have $\q_{k+1}\in\delta(\q_k,\sigma_k)$. 
We define by $\ON{Paths}_{\Aut}(s)$ the set of all paths of $\Aut$ compliant with $s$. 
We collect all finite and infinite strings that are compliant with $\Aut$ in the sets 
$\Lstar(\Aut):=\SetComp{s\in\Sigma^*}{\ON{Paths}_{\Aut}(s)\neq\emptyset}$ and
$\Lomega(\Aut):=\SetComp{s\in\Sigma^\omega}{\ON{Paths}_{\Aut}(s)\neq\emptyset}$, respectively. 
 If $\Aut$ is non-blocking, we have $\LPrefix{\Lomega(\Aut)}=\Lstar(\Aut)$. 
 If $\Aut$ is deterministic we have $|\ON{Paths}_{\Aut}(s)|=1$ for all $s\in\Lstar(\Aut)$.

\smallbreak
\noindent\textit{Finite-State Automata over Finite Words.}
Deterministic finite-state automata over finite words are typically called deterministic finite automata (DFA) and are defined by a deterministic finite state machine $\Aut$ equipped with a set of final states $F\subseteq \Q$. The DFA $(\Aut,F)$ accepts (or generates) the $*$-language $L(\Aut,F)$ which contains all finite paths of $\Aut$ which are ending in $F$. A $*$-language $L$ is called \emph{regular} iff there exists a DFA $(\Aut,F)$ which accepts $L$, i.e., $L=L(\Aut,F)$.

\smallbreak
\noindent\textit{Finite-State Automata over Infinite Words.}
For a path $\pi$, define
$\OpInf{\pi} = \set{\q\in \Q \mid \q_k=\q \mbox{ for infinitely many }k\in\mathbb{N}}$
to be the set of states visited infinitely often along $\pi$.
Let $F\subseteq \Q$ be a subset of states. 
We say that an \emph{infinite} string $s\in\Sigma^\omega$ satisfies
the \emph{B\"uchi acceptance condition} $\FcB=\Set{F}$ on $\Aut$
if there exists a path $\pi\in\ON{Paths}_\Aut(s)$ such that $\OpInf{\pi}\cap F \neq \emptyset$.
Further, let $\Fc=\Set{\tuple{G_1,R_1},\hdots,\tuple{G_m,R_m}}$ be a set, where each $G_i,R_i\subseteq \Q$, $i=1,\ldots, m$, is a subset of states. 
We say that a string $s\in\Sigma^\omega$ satisfies the \emph{Rabin acceptance condition} 
$\FcR=\Fc$ on $\Aut$ if there exists a path $\pi\in\ON{Paths}(\Aut,s)$ such that 
$\OpInf{\pi}\cap G_i \neq \emptyset$ \emph{and} $\OpInf{\pi}\cap R_i = \emptyset$ for \emph{some} $i\in [1;m]$.
It satisfies the \emph{Streett acceptance condition} $\FcS=\Fc$ if
$\OpInf{\pi}\cap G_i = \emptyset$ \emph{or} $\OpInf{\pi}\cap R_i \neq \emptyset$ for \emph{all} $i\in [1;m]$.
Rabin and Streett conditions are \emph{duals}, i.e., if $\pi$ satisfies the Rabin condition $\FcR$ it violates the Streett condition $\FcS=\FcR$. 

We call a finite state machine equipped with a Büchi, Rabin or Streett acceptance condition a Büchi, Rabin or Streett automaton, respectively. We collect all infinite strings (resp.\ paths) satisfying the specified acceptance condition $\Fc$ over $\Aut$, in the accepted language $\Lomega(\Aut,\Fc)\subseteq\Sigma^\omega$ (resp.\ in the set $\Pomega(\Aut,\Fc)\subseteq \q_0\Q^\omega$). 
An $\omega$-language $\mathcal L$ is called \emph{regular} iff it is accepted by a \emph{non-deterministic} Büchi automaton. We remark that \emph{deterministic} Rabin and \emph{deterministic} Streett automata also accept precisely the set of $\omega$-regular languages. However, this is not true for \emph{deterministic} Büchi automata which are less expressive.

\section{The Supervisor Synthesis Problem}\label{sec:SCT}

We define the supervisory controller synthesis problem following
the original formulation for $*$-languages \cite{RW87} and the subsequent extension to $\omega$-languages
in \cite{TW1991,TW1994a,TW1994b,thistle1995}.

\subsection{Problem Statement}\label{sec:SCTproblem}

Let $\ASigma$ be a finite alphabet of \emph{events}. 
A \emph{plant} is a tuple $(\Lstar_P, \Lomega_P)$, where $\Lstar_P\subseteq \ASigma^*$ is a prefix-closed regular $*$-language
and $\Lomega_P\subseteq \ASigma^\omega$ s.t.\ $\LPrefix{\Lomega_P}\subseteq\Lstar_P$ is a regular $\omega$-language. If, in addition, $\LPrefix{\Lomega_P}=\Lstar_P$, the plant is called \emph{deadlock-free}.
A \emph{specification} is a tuple $(\Lstar_S, \Lomega_S)$ where $\Lomega_S \subseteq \ASigma^\omega$ is a regular $\omega$-language and $\Lstar_S:=\LPrefix{\Lomega_S}$ is a prefix-closed regular $*$-language. 
 That is, the specification (in contrast to the plant) is by definition deadlock-free. This convention is motivated by the fact that any closed-loop system should be deadlock free in order to operate correctly, independent of other properties that should be enforced.

Intuitively, the language-tuples $(\Lstar_P, \Lomega_P)$ and $(\Lstar_S, \Lomega_S)$ capture both 
 \emph{safety} and \emph{liveness} properties. Here, $(\Lstar_P, \Lomega_P)$ models the properties the \emph{uncontrolled} plant exhibits. In contrast, $(\Lstar_S, \Lomega_S)$ restricts the behavior of the plant to a set of desired behaviors which is, by definition, deadlock-free. That is, every safe event sequence generated by the plant under control must be extendable to an infinite string additionally fulfilling the imposed liveness requirements.

\begin{remark}
 In language theory, an $\omega$-language $\Lomega$ is called a \emph{safety language} if $\OpLim{\OpPre{\Lomega}}=\Lomega$ (that is, a finite prefix of a string determines containment in a language), and a \emph{liveness language} if $\OpLim{\OpPre{\Lomega}}=\Sigma^\omega$ (that is, finite prefixes do not matter for containment).
By using this classification of languages we can interpret the language tuples $(\Lstar_P, \Lomega_P)$ and $(\Lstar_S, \Lomega_S)$ from above as follows. First, we see that both $\Lstar_P$ and $\Lstar_S$ capture the safety-part of the plant and the specification, respectively, by observing that $\OpLim{\Lstar_P}$ and $\OpLim{\Lstar_S}$ are indeed safety languages. For the liveness part, the correspondence is not as straight forward. However, as we know that any regular language can be written as the intersection of a safety and a liveness language \cite{alpern1987recognizing}, there exist (pure) liveness languages $\tilde{\Lomega}_P$ and $\tilde{\Lomega}_S$ s.t.\ $\Lomega_P=\tilde{\Lomega}_P\cap\OpLim{\Lstar_P}$ and $\Lomega_S=\tilde{\Lomega}_S\cap\OpLim{\Lstar_S}=\tilde{\Lomega}_S\cap\OpLim{\OpPre{\Lomega_S}}$.
Hence, in the context of SCT, the $\omega$-languages $\Lomega_P$ and $\Lomega_S$ capture both safety and liveness properties. 

For any interesting instance of the SCT control problem over infinite strings, we require $\Lomega_P\subsetneq\OpLim{\Lstar_P}$ and $\Lomega_S\subsetneq\OpLim{\Lstar_S}$, that is, both the plant and the specification languages $\Lomega_P$ and $\Lomega_S$ do not only contain a safety property (i.e., $\OpLim{\Lstar_P}$ and $\OpLim{\Lstar_S}$) but also a non-trivial liveness property (captured by $\Lomega_P\setminus \OpLim{\Lstar_P}$ and $\Lomega_S\setminus \OpLim{\Lstar_S}$, respectively). We refer the reader to 
\cite{Moor_Report_omegaSCT} for an accessible discussion of this topic.
\end{remark}

Given the finite alphabet $\ASigma$, its subset $\ASigmac\subseteq\ASigma$ denotes all events the controller can prevent the plant from executing, while the set $\ASigmauc\subseteq\ASigma$ denotes events that cannot be prevented by the controller. We typically require that $\ASigmac$ and $\ASigmauc$ form a partition of $\ASigma$, i.e., $\ASigma= \ASigmac\dotcup\ASigmauc$.

Further, a \emph{control pattern} $\gamma$ is a subset of $\Sigma$ containing $\ASigmauc$. We collect all control pattern in the set 
 $\Gamma\D\{\,\gamma\subseteq\ASigma\,|\,\ASigmauc\subseteq\gamma\,\}$.
Given this set, a \emph{(string-based) supervisor} is defined as a map $f:\Sigma^*\rightarrow\Gamma$ that maps each (finite) past event sequence $s\in\Sigma^*$ to a control pattern $f(s)\in\Gamma$.
The control pattern specifies the set of enabled successor events after the occurrence of $s$. The definition of control patterns ensures that
uncontrollable events are always enabled.
A word $s\in\Sigma^*$ is called \emph{consistent} with $f$ if for all $\sigma\in\Sigma$ and 
$t\sigma\in\LPrefix{s}$, it holds that $\sigma\in f(t)$. 
We write $\Lstar_f$ for the set of all words consistent with $f$ and
define $\Lomega_f:=\OpLim{\Lstar_f}$.

With these definitions, the supervisor synthesis problem can be formally stated as follows.

\begin{problem}[String-Based Supervisor Synthesis]\label{prob:SCT}
 Given an alphabet $\ASigma= \ASigmac\dotcup\ASigmauc$, a plant model $(\Lstar_P,\Lomega_P)$, where $\Lomega_P\subseteq\ASigma^\omega$ and $\LPrefix{\Lomega_P}\subseteq\Lstar_P\subseteq \ASigma^*$ are regular languages, 
 and a regular specification language%
 \footnote{ As $\Lstar_S:=\LPrefix{\Lomega_S}$, the language $\Lstar_S$ is uniquely determined by $\Lomega_S$ and therefore omitted from the problem description.}
 $\Lomega_S\subseteq \ASigma^\omega$, 
 synthesize, if possible, a \emph{string-based supervisor} $f:\Sigma^*\rightarrow \Gamma$ s.t.\
\begin{subequations}\label{equ:f}
 \begin{compactenum}[(i)]
 \item the closed-loop satisfies the specification, i.e.,
 \begin{equation}\label{equ:f:contain}
  \emptyset\subsetneq\Lomega_{f}\cap\Lomega_P \subseteq \Lomega_S
 \end{equation}
 \item the plant and the supervisor are \emph{non-conflicting}, i.e.,
 \begin{equation}\label{equ:f:nonconf}
  \Lstar_f\cap\Lstar_P \subseteq \LPrefix{\Lomega_f\cap\Lomega_P},
 \end{equation}
\end{compactenum}
\end{subequations}
or determine that no such supervisor exists.
A string-based supervisor ${f}$ solves the synthesis problem over $((\Lstar_P,\Lomega_P),\Lomega_S)$
if it satisfies \eqref{equ:f:contain} and \eqref{equ:f:nonconf}.
\hfill$\triangleleft$
\end{problem}
The constraint \eqref{equ:f:nonconf} ensures that the plant is always able to generate events allowed by $f$ 
s.t.\ it ultimately generates a word in the language $\Lomega_P$. 
Then, by \eqref{equ:f:contain}, all such generated words must be contained in the specification $\Lomega_S$.

\subsection{A Special Case: Terminating Processes}
\label{sec:sct:terminating}
Given that many readers might be more familiar with the supervisory controller synthesis problem for \emph{terminating} processes, we first recall its special case of \REFproblem{prob:SCT} and a standard algorithmic solution before discussing automata realizations for solving \REFproblem{prob:SCT} for $\omega$-regular input parameters.

Within the basic setting of supervisory controller synthesis for \emph{terminating processes}, the languages $\Lomega_P$ and $\Lomega_S$ are non-prefix closed regular $*$-languages, rather than regular $\omega$-languages. They are typically called the \emph{marked} language and denoted by $\Lstar_{mP}$ and $\Lstar_{mS}$, respectively. In this setting, it is well known that the tuples $(\Lstar_P, \Lstar_{mP})$ and  $(\Lstar_S, \Lstar_{mS})$ can be represented by deterministic finite automata (DFA) denoted by $(M_P,F_P)$ and $(M_S,F_S)$, s.t.\ $\Lstar_P:=\Lstar(M_P)$, $\Lstar_{mP}=\Lstar(M_P,F_P)$, $\Lstar_S:=\Lstar(M_S)$, and $\Lstar_{mS}=\Lstar(M_S,F_S)$. I.e., $\Lstar_P$ and $\Lstar_S$ collect all \emph{finite} strings generated by $M_P$ and $M_S$, respectively, when starting form the initial state and ending in any other state $x\in X$. Similarly, $\Lstar_{mP}$ and $\Lstar_{mS}$ collect all \emph{finite} strings generated by $M_P$ and $M_S$, respectively, when starting form the initial state and ending in a \emph{marked} state. As in the $\omega$-language case, it is further assume that $(M_S,F_S)$ is non-blocking, i.e., $\Lstar_S=\LPrefix{\Lstar_{mS}}$.
In addition, one typically requires that $\Lstar_{mS}$ is closed w.r.t.\ $\Lstar_{mP}$.

In the standard version of the supervisory control problem over terminating processes (see e.g.\ \cite{CassandrasLafortune3rd}, p.184) one usually refers to the intersections $\Lstar_{f}\cap\Lstar_{P}$ and  $\Lstar_{f}\cap\Lstar_{mP}$ as the unmarked and marked language of the closed-loop (i.e., the plant under supervision), usually denoted by $\Lstar(f/P)$ and $\Lstar_m(f/P)$. Then the standard supervisory control problem\footnote{Note that the usual controllablity requirement is hidden in the definition of $f$ to always enable uncontrollable events. } asks for a supervisor $f:\Sigma^*\fun \Gamma$, s.t.\ the closed-loop is non-blocking, i.e., $\Lstar(f/P)=\LPrefix{\Lstar_m(f/P)}$ and $\LPrefix{\Lstar_m(f/P)}$ is the maximal language (with the above properties) contained in the marked specification language $\Lstar_{mS}$. 

Matching these two requirements to \REFproblem{prob:SCT} we see that \eqref{equ:f:nonconf} corresponds to the non-blocking requirement, while \eqref{equ:f:contain} requires containment in the marked specification language. The fact that the classical supervisory synthesis problem asks for a maximal solution stems from the fact that such a maximal solution does uniquely exist for \emph{terminating} processes. It is worth noting, that this is in general not true for \emph{non-terminating} processes, and hence omitted from \REFproblem{prob:SCT}.

One standard algorithmic solution to the outlined supervisor synthesis problem (see \cite{CassandrasLafortune3rd}, p.186) combines both DFA realizations $(M_P,F_P)$ and $(M_S,F_S)$ into a single DFA $(M,F)$. This process is sometimes called \enquote{plantification of the specification}.  
This construction essentially extends the DFA $(M_S,F_S)$ into a complete DFA $(\widetilde{M}_S,F_S)$ whose unmarked language is unrestricted, i.e., $\Lstar(\widetilde{M}_S)=\Sigma^*$ and then takes a normal automata product of $(M_P,F_P)$ and $(\widetilde{M}_S,F_S)$ to obtain the synthesis automaton $(M,F)$ where a state $(q,p)$ is marked (i.e., contained in $F$) if $q\in F_P$ and $p\in F_S$. 
Intuitively, solving \REFproblem{prob:SCT} reduces to strategically disabling controllable transition in $M$ s.t.\ states in $F$ always remain reachable, which ensures safety, non-blockingness and controllability. It is proven in (see \cite{CassandrasLafortune3rd}, p.186) that this construction is indeed correct and results in the desired supervisor. 

As an example, consider the DFA's $(M_P,F_P)$ and $(M_S,F_S)$ depicted in \REFfig{fig:exp:MpMs} (top) and (middle) and the completed DFA  $(\widetilde{M}_S,F_S)$ in \REFfig{fig:exp:MpMs} (bottom). Their product results in the synthesis automaton $(M,F)$ depicted in \REFfig{fig:exp:MF}. Assuming that all events are controllable, the resulting supervisor would disable events $d$ and $c$ in $s_3$.

\begin{figure}
 \begin{tikzpicture}[auto]
    \begin{footnotesize}
          \node (name) at (-2,0) {$\Aut_P$:};
          \node (init) at (2,0.3) {};
          \node[state] (p1) at (0,0) {$p_1$};
          \node[state] (p2) at (0,-1) {$p_2$};
          \node[state] (p3) at (2,-0.5) {$p_3$};
          \node[state,accepting,sws-blue] (p4) at (4,0) {$p_4$};
          \node[state,accepting,sws-blue] (p5) at (4,-1) {$p_5$};
          
          \SFSAutomatEdge{init}{}{p3}{}{}  
          \SFSAutomatEdge{p3}{a}{p2}{}{swap}
          \SFSAutomatEdge{p2}{d}{p3}{bend right}{swap}
          \SFSAutomatEdge{p3}{d}{p1}{}{swap}
          \SFSAutomatEdge{p3}{c}{p4}{}{swap}
          \SFSAutomatEdge{p4}{d}{p3}{bend right}{swap}
          \SFSAutomatEdge{p3}{b}{p5}{}{pos=0.7}
          \SFSAutomatEdge{p5}{d}{p3}{bend left}{}       
  \end{footnotesize}
 \end{tikzpicture}
 
  \begin{tikzpicture}[auto]
    \begin{footnotesize}
          \node (name) at (-2,0) {$\Aut_S$:};
          \node (init) at (2,0.3) {};
          \node[state,accepting,sws-red] (p2) at (0,-0.5) {$q_2$};
          \node[state] (p3) at (2,-0.5) {$q_3$};
          \node[state,accepting,sws-red] (p5) at (4,-0.5) {$q_5$};
          
          \SFSAutomatEdge{init}{}{p3}{}{}  
          \SFSAutomatEdge{p3}{a}{p2}{bend right}{swap}
          \SFSAutomatEdge{p2}{d}{p3}{bend right}{swap}
          \SFSAutomatEdge{p3}{b}{p5}{bend left}{}
          \SFSAutomatEdge{p5}{d}{p3}{bend left}{}  
          \SFSAutomatEdge{p3}{d}{p3}{loop below}{} 
  \end{footnotesize}
 \end{tikzpicture}

  \begin{tikzpicture}[auto]
    \begin{footnotesize}
          \node (name) at (-2,0) {$\widetilde{\Aut}_S$:};
          \node (init) at (2,0.3) {};
          \node[state,accepting,sws-red] (p2) at (0,-0.5) {$q_2$};
          \node[state] (p3) at (2,-0.5) {$q_3$};
          \node[state,accepting,sws-red] (p5) at (4,-0.5) {$q_5$};
          \node[state] (dum) at (2,-2) {$\bot$};
          
          \SFSAutomatEdge{init}{}{p3}{}{}  
          \SFSAutomatEdge{p3}{a}{p2}{bend right}{swap}
          \SFSAutomatEdge{p2}{d}{p3}{bend right}{swap}
          \SFSAutomatEdge{p3}{b}{p5}{bend left}{}
          \SFSAutomatEdge{p5}{d}{p3}{bend left}{}  
          \SFSAutomatEdge{p3}{d}{p3}{loop below}{} 
          
          \SFSAutomatEdge{p2}{a,b,c}{dum}{bend right}{sloped,swap}
          \SFSAutomatEdge{p5}{a,b,c}{dum}{bend left}{sloped,swap}
          \SFSAutomatEdge{p3}{c}{dum}{bend left}{}
          \SFSAutomatEdge{dum}{a,b,c,d}{dum}{loop below}{}
  \end{footnotesize}
 \end{tikzpicture}
 
 \caption{Example automata for the construction of synthesis automta. For terminating processes, we interpret them as DFA $(M_P,F_P)$, $(M_S,F_S)$ and $(\widetilde{M}_S,F_S)$ and for non-terminating processes we interpret them as Büchi automata $(M_P,\Fc_P)$, $(M_S,\Fc_S)$ and $(\widetilde{M}_S,\Fc_S)$ with $\Fc_i=\set{F_i},~i\in\set{P,S}$. States in the set $F_i$ are indicated by double circles.}\label{fig:exp:MpMs}
\end{figure}
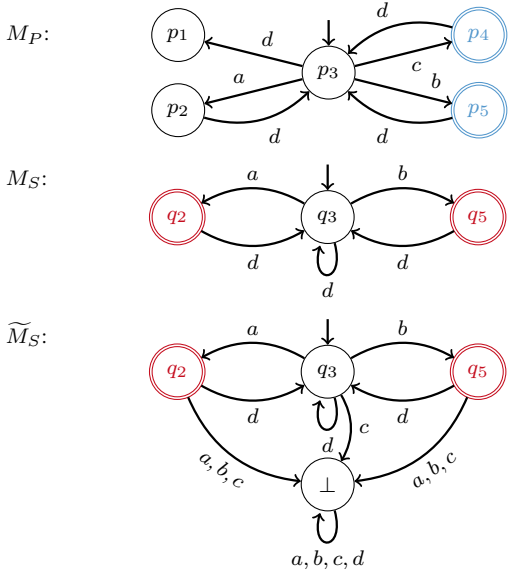

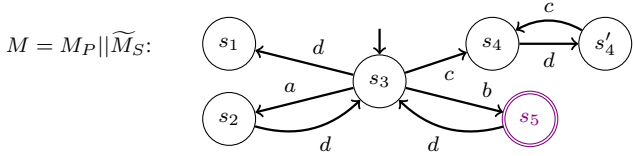
\begin{figure}
 \begin{tikzpicture}[auto]
    \begin{footnotesize}
          \node (name) at (-2,0) {$\Aut=\Aut_P||\widetilde{\Aut}_S$:};
          \node (init) at (2,0.3) {};
          \node[state] (p1) at (0,0) {$s_1$};
          \node[state] (p2) at (0,-1) {$s_2$};
          \node[state] (p3) at (2,-0.5) {$s_3$};
          \node[state] (p4) at (3.5,0) {$s_4$};
          \node[state] (p4p) at (5,0) {$s'_4$};
          \node[state,accepting,violet] (p5) at (4,-1) {$s_5$};
          
          \SFSAutomatEdge{init}{}{p3}{}{}  
          \SFSAutomatEdge{p3}{a}{p2}{}{swap}
          \SFSAutomatEdge{p2}{d}{p3}{bend right}{swap}
          \SFSAutomatEdge{p3}{d}{p1}{}{swap}
          \SFSAutomatEdge{p3}{c}{p4}{}{swap}
          \SFSAutomatEdge{p4}{d}{p4p}{}{swap}
          \SFSAutomatEdge{p4p}{c}{p4}{bend right}{swap}
          \SFSAutomatEdge{p3}{b}{p5}{}{pos=0.7}
          \SFSAutomatEdge{p5}{d}{p3}{bend left}{}       
  \end{footnotesize}
 \end{tikzpicture}
 
\caption{Synthesis automaton $(M,F)$ for the supervisory controller synthesis problem over \emph{terminating processes} depicted in \REFfig{fig:exp:MpMs}. States in $F$ are indicated by violet douple circled states. }\label{fig:exp:MF}
\end{figure}

\subsection{Automata Representations for Supervisor Synthesis}\label{sec:real}
The intuition behind the automata realization for solving \REFproblem{prob:SCT} directly carries over from terminating processes (as discussed in \REFsec{sec:sct:terminating}) to non-terminating processes. 
If $\Lomega_P$ and $\Lomega_S$ are indeed non-(topologically) closed regular $\omega$-languages, the language tuples $(\Lstar_P, \Lomega_P)$ and  $(\Lstar_S, \Lomega_S)$ are realized by 
B\"uchi, Rabin or Streett automata, instead of DFA's. I.e., as all involved languages are regular, there exist automata $(M_P,\Fc_P)$ $(M_S,\Fc_S)$ s.t.\ $\Lstar_P$ and $\Lstar_S$ collect all \emph{finite} strings compliant with $M_P$ and $M_S$, and $\Lomega_P$ and $\Lomega_S$ collect all \emph{infinite} strings that are compliant with $M_P$ and $M_S$, respectively, and, in addition, satisfy the given acceptance condition $\Fc_P$ and $\Fc_S$ over $M_P$ and $M_S$, respectively.
As any $\omega$-regular language is accepted (or generated) by a \emph{deterministic} Rabin or Streett automaton, we can, without loss of generality, follow \cite{thistle1995} and assume that $(M_P,\Fc_P)$ is a deterministic Streett automaton, while $(M_S,\Fc_S)$ is a deterministic Rabin automaton.

Similar to the $*$-language case, both $\omega$-automata can be combined to a single machine $\Aut$ while using the premisse of \REFproblem{prob:SCT} that $(M_S,\Fc_S)$ is non-blocking, i.e., $\Lstar(M_S)=\OpPre{\Lomega(M_S,\Fc_S)}$. The major difference here is that $\omega$-regular acceptance conditions do not require \enquote{synchonization}, that is, plant and specification markings need not be reached \emph{at the same time}. An infinite path over $\Aut$ can fulfill the two different acceptance conditions $\Fc_P$ and $\Fc_S$ (for example visiting a state in $\Fc_P$ \emph{and} a state in $\Fc_S$ infinitely often in the case of B\"uchi conditions) without requiring that both conditions are met \emph{at the same time}. This is the reason why SCT for \emph{non-terminating processes} does not require that the specification is closed w.r.t.\ the plant language. Further, this implies that the product of the $\omega$-automata $(\Aut_P,\Fc_P)$ and $(\Aut_S,\Fc_S)$ results in a single automaton $\Aut$ equipped with \emph{two} sets of marked states  $\Fc'_P$ and $\Fc'_S$ over $\Aut$. This is formalized in the following definition.

\begin{definition}\label{def:automataprodcut}
 Let $\Aut_P=(\Q_P,\,\Sigma,\,\delta_P,\, \q_{P,0})$ and $\Aut_S=(\Q_S,\,\Sigma,\,\delta_S,\, \q_{S,0})$ be two deterministic state machines over the same alphabet $\Sigma$ with  $\Fc_P=\Set{\tuple{G_{P,1},R_{P,1}},\hdots,\tuple{G_{mP},R_{mP}}}$ and $\Fc_S=\Set{\tuple{G_{S,1},R_{S,1}},\hdots,\tuple{G_{S,n},R_{S,n}}}$ being a Streett and a Rabin acceptance condition over $\Q_P$ and $\Q_S$, respectively. 
 Further, let $\widetilde{\Aut}_S=(\Q_S\cup\{\bot\},\,\Sigma,\,\widetilde{\delta}_S,\, \q_{S,0})$ be the completion of $\Aut_S$, that is, for all $\q\in\Q_S$ and $\sigma\in\Sigma$ holds that $\bot=\widetilde{\delta}(\q,\sigma)$ iff $\delta(\q,\sigma)$ is undefined, and $\bot=\widetilde{\delta}(\bot,\sigma)$ for all $\sigma\in\Sigma$.
 
 Then we define the extended product of $(\Aut_P,\Fc_P)$ and $(\Aut_S,\Fc_S)$ as the tuple $(\Aut,\Fc'_P,\Fc'_S)$ with
 \begin{compactitem}
  \item $\Aut=(\Q,\,\Sigma,\,\delta,\, \q_{0})$ s.t.\
 $\Q=\Q_P\times\widetilde{\Q}_S$, $q_0=(\q_{P,0},\q_{S,0})$, and 
 $(q',p')=\delta((q,p),\sigma)$ iff $q'=\delta_P(q,\sigma)$ and $p'=\widetilde{\delta}_S(p,\sigma)$,
 \item $\Fc'_P=\Set{\tuple{G'_{P,1},R'_{P,1}},\hdots,\tuple{G'_{mP},R'_{mP}}}$ s.t.\ $G'_{P,i}=\SetComp{(p,q)\in\Q}{p\in G_{P,i}}$ and $R'_{P,i}=\SetComp{(p,q)\in\Q}{p\in R_{P,i}}$, and 
  \item $\Fc'_S=\Set{\tuple{G'_{S,1},R'_{S,1}},\hdots,\tuple{G'_{mS},R'_{mS}}}$ s.t.\ $G'_{S,i}=\SetComp{(p,q)\in\Q}{q\in G_{S,i}}$ and $R'_{S,i}=\SetComp{(p,q)\in\Q}{q\in R_{S,i}}$.
 \end{compactitem}
\end{definition}

Due to the usual properties of automata completion and product, we have the following observations.

\begin{lemma}\label{lem:automataprodcut}
 Given the premisses of \REFdef{def:automataprodcut}, it holds that
 \begin{inparaenum}[(a)]
  \item $\Lstar(\Aut_P)=\Lstar(\Aut)$, 
   \item $\Lomega(\Aut_P,\Fc_P)=\Lomega(\Aut,\Fc'_P)$, 
  \item $\Lomega(\Aut_S,\Fc_S)\cap\Lomega(\Aut)=\Lomega(\Aut,\Fc'_S)$ and
  \item $\Aut$ is deterministic. 
 \end{inparaenum}
\end{lemma}

\begin{proof}
 \begin{inparaenum}[(a)]
 First, recall that completing the finite state machine $\Aut_S$ into $\widetilde{\Aut}_S$ implies $\Lstar(\widetilde{\Aut})=\Sigma^*$ and $\Lomega(\widetilde{\Aut})=\Sigma^\omega$. Further, observe that the construction of $\Aut$ is the usual product of $\Aut_P$ and $\widetilde{\Aut}_S$, implying that $\Aut$ is deterministic (i.e., (d) holds), $\Lstar(\Aut)=\Lstar(\Aut_P)\cap\Lstar(\widetilde{\Aut}_S)=\Lstar(\Aut_P)\cap\Sigma^*=\Lstar(\Aut_P)$ and $\Lomega(\Aut)=\Lomega(\Aut_P)\cap\Lomega(\widetilde{\Aut}_S)=\Lomega(\Aut_P)\cap\Sigma^\omega=\Lomega(\Aut_P)$. 
 With this (a) directly holds. Further, $\Lomega(\Aut)=\Lomega(\Aut_P)$ and the construction of $\Fc'_P$ from $\Fc_P$ implies that a string $\omega\in\Lomega(\Aut)$ fulfills $\Fc'_P$ iff it fulfills $\Fc_P$, hence (b) holds. Similarly, the construction of $\Fc'_S$ from $\Fc_S$ implies that a string $\omega\in\Lomega(\Aut)$ fulfills $\Fc'_S$ iff it fulfills $\Fc_S$, implying (c).
 \end{inparaenum}
\end{proof}

In order to use the automaton $(\Aut,\Fc'_P,\Fc'_S)$ to solve \REFproblem{prob:SCT} we additionally need that distinct transitions in $\Aut$ carry distinct labels, i.e.\ for any $\sigma,\sigma'\in\Sigma$ and $x\in X$ we have that $\delta(x,\sigma)=\delta(x,\sigma')$ implies $\sigma=\sigma'$. This can be enforced by the following construction.

\begin{definition}\label{def:distincttransition}
 Let $(\Aut,\Fc_P,\Fc_S)$ be a deterministic extended product automaton. Then we define its \emph{distinct transition version} as the tuple $(\Aut',\Fc'_P,\Fc'_S)$ s.t.\  \\
 \begin{inparaenum}[(a)]
  \item $\Aut':=(\Q\times\Sigma,\,\Sigma,\,\delta',\, \q_{0})$ with 
 $(\q',\sigma)=\delta'(\q_0,\sigma)$ iff $\q'=\delta(\q_0,\sigma)$ and for all $\overline{\sigma}\in\Sigma$ holds that
 $(\q',\sigma)=\delta'((\q,\overline{\sigma}),\sigma)$ iff $\q'=\delta(\q,\sigma)$,\\ 
 \item $\Fc'_P:=\SetComp{(\q,\sigma)\in \Q'}{\q\in\Fc_P}\cup (\set{\q_0}\cap\Fc_P)$, and \\
 \item $\Fc'_S:=\SetComp{(\q,\sigma)\in \Q'}{\q\in\Fc_S}\cup (\set{\q_0}\cap\Fc_S)$.
 \end{inparaenum}
\end{definition}

The above construction is trivially language preserving, i.e., we have the following lemma.

\begin{lemma}\label{lem:distincttransition}
 Given the premises of \REFdef{def:distincttransition}, it holds that
 \begin{inparaenum}[(a)]
  \item $\Lstar(\Aut)=\Lstar(\Aut')$, 
   \item $\Lomega(\Aut,\Fc_P)=\Lomega(\Aut',\Fc'_P)$, 
  \item $\Lomega(\Aut,\Fc_S)=\Lomega(\Aut',\Fc'_S)$,
  \item $\Aut'$ is deterministic and 
  \item distinct transitions in $\Aut'$ carry distinct labels.
 \end{inparaenum}
\end{lemma}

\begin{proof}
 We first note that determinicity of $\Aut'$ follows from the fact that $\Aut$ is deterministic and we are only further splitting transitions, and not merging them. In addition $\Aut'$ has a single initial state. As both $\Aut$ and $\Aut'$ are deterministic, we further see that for every finite string $\alpha\in\Sigma^*$ a unique state is reached s.t.\ $\delta(\q_0,\alpha)=\q$ iff $\delta'(\q_0,\alpha)=(\q,\ON{last}(\alpha))$. This immediately shows (a) and implies from the definition of $\Fc'_P$ and $\Fc'_S$ from $\Fc_P$ and $\Fc_S$ that (b) and (c) also hold. (e) immediately follows from the construction.
\end{proof}

Summarizing the discussion above and recalling that any regular $\omega$-language is realizabel by a deterministic Streett or Rabin automaton, we see that every input $((\Lstar_P,\Lomega_P),\Lomega_S)$ to \REFproblem{prob:SCT} can be realized by a so called \emph{Streett/Rabin supervisor synthesis automaton} $(\Aut,\FcS_P,\FcR_S)$. This is summarized in the following proposition.

\begin{proposition}\label{prop:StreetRabinAut}
Let $\Lstar_P\subseteq \Sigma^*$, $\Lomega_P,\Lomega_S\subseteq \Sigma^\omega$ be regular languages and $((\Lstar_P,\Lomega_P), \Lomega_S)$ an input to \REFproblem{prob:SCT}. Then there exists a finite state machine $\Aut=(\Q,\,\Sigma,\,\delta,\, \q_0)$, a Street condition $\FcS_P$ over $\Aut$ and a Rabin condition $\FcR_S$ over $\Aut$, s.t.\ 
  \begin{inparaenum}[(a)]
   \item $\Lstar_P=\Lstar(\Aut)$, 
   \item $\Lomega_P=\Lomega(\Aut,\FcS_P)$,  
  \item $\Lomega_S\cap\Lomega(\Aut)=\Lomega(\Aut,\FcR_S)$,
  \item $\Aut$ is deterministic, and
  \item distinct transitions in $\Aut$ carry distinct labels, i.e.\ for any $\sigma,\sigma'\in\Sigma$ and $x\in X$ we have that $\delta(x,\sigma)=\delta(x,\sigma')$ implies $\sigma=\sigma'$.
 \end{inparaenum}
\end{proposition}

\begin{proof}
 Let $((\Lstar_P,\Lomega_P),\Lomega_S)$ be an arbitrary regular input to \REFproblem{prob:SCT} over the alphabet $\Sigma$. As all languages are regular, there exists a \emph{deterministic} Streett automaton $(\Aut_P,\Fc_P)$ and a \emph{deterministic} Rabin automaton $(\Aut_S,\Fc_S)$ s.t.\ 
 $\Lstar_P=\Lstar(\Aut_P)$, $\Lomega_P=\Lomega(\Aut_P,\Fc_P)$, $\Lstar_S:=\OpPre{\Lomega_S}=\Lstar(\Aut_S)$, $\Lomega_S=\Lomega(\Aut_S,\Fc_S)$. Now first applying \REFdef{def:automataprodcut} to $(\Aut_P,\Fc_P)$ and $(\Aut_S,\Fc_S)$ yields the extended product automaton $(\Aut^\times,\Fc^\times_P,\Fc^\times_S)$. Applying \REFdef{def:distincttransition} to $(\Aut^\times,\Fc^\times_P,\Fc^\times_S)$ yields its distrinct transition version $(\Aut,\FcS_P,\FcR_S)$. With this, conditions (a)-(e) in \REFprop{prop:StreetRabinAut} immediately follow from \REFlem{lem:automataprodcut} and \REFlem{lem:distincttransition}.
\end{proof}

\begin{definition}\label{def:StreetRabinAut}
  If conditions (a)-(e) in \REFprop{prop:StreetRabinAut} are fulfilled, we call the tuple $(\Aut,\FcS_P,\FcR_S)$ the \emph{Streett/Rabin supervisor synthesis automaton} realizing $((\Lstar_P,\Lomega_P),\Lomega_S)$.
\end{definition}

To illustrate the similarities and differences of the Streett/Rabin supervisor synthesis automaton with the synthesis automaton for terminating processes, we revisit the previous example and interpret the automata in \REFfig{fig:exp:MpMs} as $\omega$-automata. As a deterministic Büchi automaton is a special case of both a Streett and a Rabin automaton\footnote{Let $(M,\set{F})$ be a B\"uchi automaton. Then this automaton is equivalent to the Rabin automaton $(M,\tuple{(F,\emptyset)})$ and the Streett automaton $(M,\tuple{(\Q,F)})$.}, we simply interpret $(M_P,\Fc_P)$, $(M_S,\Fc_S)$ and $(\widetilde{M}_S,\Fc_S)$ as Büchi automata with $\Fc_i=\set{F_i},~i\in\set{P,S}$. Now following the construction in \REFdef{def:automataprodcut} we obtain the Streett/Rabin supervisor synthesis automaton depicted in \REFfig{fig:examle:MPS} which already has distinct transitions. We see that the  Streett/Rabin supervisor synthesis automaton in \REFfig{fig:examle:MPS} has both deadlocks (see state $s_1$) and livelocks (the loop between $s_4$ and $s'_4$ which does not allow to reach a specification marking infinitely often) which need to be prevented by the supervisor to imply safety of the closed loop. To additionally fulfill the liveness constrains imposed by the specification, the supervisor needs to make sure that every 
 infinite trace of the controlled system visits red/violet states infinitely often if it visits blue/violet states infinitely often. For this example this trivially holds for every safe path. We will see in \REFsec{sec:example} that this is usually not the case.

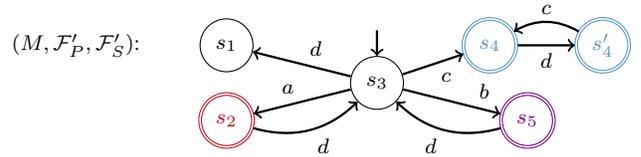
\begin{figure}
 \begin{tikzpicture}[auto]
    \begin{footnotesize}
          \node (name) at (-2,0) {$(\Aut,\Fc'_P,\Fc'_S)$:};
          \node (init) at (2,0.3) {};
          \node[state] (p1) at (0,0) {$s_1$};
          \node[state,accepting,sws-red] (p2) at (0,-1) {$s_2$};
          \node[state] (p3) at (2,-0.5) {$s_3$};
          \node[state,accepting,sws-blue] (p4) at (3.5,0) {$s_4$};
          \node[state,accepting,sws-blue] (p4p) at (5,0) {$s'_4$};
          \node[state,accepting,violet] (p5) at (4,-1) {$s_5$};
          
          \SFSAutomatEdge{init}{}{p3}{}{}  
          \SFSAutomatEdge{p3}{a}{p2}{}{swap}
          \SFSAutomatEdge{p2}{d}{p3}{bend right}{swap}
          \SFSAutomatEdge{p3}{d}{p1}{}{swap}
          \SFSAutomatEdge{p3}{c}{p4}{}{swap}
          \SFSAutomatEdge{p4}{d}{p4p}{}{swap}
          \SFSAutomatEdge{p4p}{c}{p4}{bend right}{swap}
          \SFSAutomatEdge{p3}{b}{p5}{}{pos=0.7}
          \SFSAutomatEdge{p5}{d}{p3}{bend left}{}       
  \end{footnotesize}
 \end{tikzpicture}
 
 \caption{Streett/Rabin supervisor synthesis automaton $(\Aut,\Fc'_P,\Fc'_S)$ resulting from the Büchi automata $(M_P,\Fc_P)$, $(M_S,\Fc_S)$ and $(\widetilde{M}_S,\Fc_S)$ depicted in \REFfig{fig:exp:MpMs} with $\Fc'_P=\set{s_4,s_4',s_5}$ (blue,violet) and $\Fc'_S=\set{s_2,s_5}$ (red,violet).}\label{fig:examle:MPS}
\end{figure}

\subsection{Path-Based Supervisor Synthesis}

We now formalize the outlined connection between effective algorithms to solve \REFproblem{prob:SCT} via a Streett/Rabin supervisor synthesis automaton $(\Aut,\FcS_P,\FcR_S)$ realizing $((\Lstar_P,\Lomega_P),\Lomega_S)$, and \REFproblem{prob:SCT}.
This is done via a re-formulation of \REFproblem{prob:SCT} in terms of $(\Aut,\FcS_P,\FcR_S)$ and so called 
\emph{path-based} supervisors which is proven to be equivalent to \REFproblem{prob:SCT}. 

To this end, we define a \emph{path-based supervisor} as a map $\check{f}:\Pstar(\Aut)\fun \Gamma$.
A path $\pi$ over $\Aut$ is called \emph{consistent} with $\check{f}$ 
if for all $\q,\q'\in \Q$ and $\nu\in\Q^*$ s.t.\ $\nu \q \q'\in\LPrefix{\pi}$, there exists an event $\sigma\in \check{f}(\nu \q)$ 
s.t.\ $\q'=\delta(\q,\sigma)$. 
Let $\Pstar(\Aut,\check{f})$ be the set of all paths of $\Aut$ consistent with $\check{f}$ and 
define $\Pomega(\Aut,\check{f}):=\OpLim{\Pstar(\Aut,\check{f})}$. 
We can now re-state \REFproblem{prob:SCT} into the following path-based supervisor synthesis problem.

\begin{problem}[Path-based Supervision]\label{prob:SCT_P}
Given a Streett/Rabin supervisor synthesis automaton $(\Aut,\FcS_P,\FcR_S)$, synthesize, if possible,
a path-based supervisor $\check{f}:\Pstar(\Aut)\rightarrow \Gamma$ s.t.\
\begin{subequations}\label{equ:Pf}
\begin{align}
 &\emptyset\subsetneq\Pomega(\Aut,\check{f})\cap \Pomega(\Aut,\FcS_P)\subseteq\Pomega(\Aut,\FcR_S),~\text{and}\label{equ:Pf:a}\\
 &\Pstar(\Aut,\check{f})\subseteq\LPrefix{\Pomega(\Aut,\check{f})\cap\Pomega(\Aut,\FcS_P)},\label{equ:Pf:b}
\end{align}
\end{subequations}
or determine that no such supervisor exists.
A path-based supervisor $\check{f}$ solves the synthesis problem over  $(\Aut,\FcS_P,\FcR_S)$
if it satisfies \eqref{equ:Pf:a} and \eqref{equ:Pf:b}.
\end{problem}

The structure of the finite state machine $\Aut$ ensures that there is a one-to-one correspondence between a word in $\Lstar_P=\Lstar(\Aut)$ 
and its unique path $\pi=\ON{Paths}_{\Aut}(s)$ over $\Aut$. 
Further, as transition labels are unique in $\Aut$ (see \REFprop{prop:StreetRabinAut} (e)), there is 
also a unique word $s$ associated with a path $\pi$ over $\Aut$. 
With these observations, we can show that \REFproblem{prob:SCT} and \ref{prob:SCT_P} are indeed equivalent, 
as summarized by \REFthm{thm:SCTvsSCTP}.

\begin{theorem}\label{thm:SCTvsSCTP}
Let $\Lstar_P\subseteq \Sigma^*$, $\Lomega_P,\Lomega_S\subseteq \Sigma^\omega$ be regular languages.
Let $(\Aut,\FcS_P,\FcR_S)$ be a realizing Streett/Rabin supervisor synthesis automaton for the input $((\Lstar_P,\Lomega_P), \Lomega_S)$ to \REFproblem{prob:SCT}. Further, let $f:\Sigma^*\rightarrow \Gamma$ and $\check{f}:\Pstar(\Aut)\rightarrow \Gamma$ be a string- and a path-based supervisor, respectively, s.t.\
 \begin{equation}\label{equ:fvscheckf}
  \AllQ{s\in\Lstar(\Aut)}{f(s)=\check{f}(\ON{Paths}_{\Aut}(s))}.
 \end{equation}
Then $f$ solves the synthesis problem over $((\Lstar_P,\Lomega_P),\Lomega_S)$
iff 
$\check{f}$ solves the synthesis problem over $(\Aut,\FcS_P,\FcR_S)$.

\end{theorem}

\begin{proof}
We show both directions separately.\\
\begin{inparaitem}[$\blacktriangleright$]
 \item \enquote{$\Rightarrow$}:
 Fix $f$ s.t.\ \eqref{equ:f} holds, and $\check{f}$ s.t. \eqref{equ:fvscheckf} holds. \\
\begin{inparaitem}[$\triangleright$]
 \item We first show that $\Pomega(\Aut,\check{f})\cap \Pomega(\Aut,\FcS_P)\neq\emptyset$. To this end, recall that \eqref{equ:f:contain} holds, i.e., $\Lomega_f\cap\Lomega_P\neq\emptyset$. This implies that there exists $s\in\Lomega_P$ s.t.\ for all $\sigma\in\Sigma$ and $t\sigma\in\LPrefix{s}$ holds that $\sigma\in f(t)$. As $\Aut$ is deterministic, we have $|\ON{Paths}_{\Aut}(s)|=1$ and define $\pi:=\ON{Paths}_{\Aut}(s)$. With $\Lomega_P=\Lomega(\Aut,\FcS_P)$ we have $\pi\in\Pomega(\Aut,\FcS_P)$. As $\Aut$ is deterministic, we have $\ON{Paths}_{\Aut}(t\sigma)\in\LPrefix{\pi}\in\Pstar(\Aut)$ for all $t\sigma\in\LPrefix{s}$. This implies $\pi\in\Pomega(\Aut,\check{f})$ from \eqref{equ:fvscheckf}, i.e.,  $\pi\in\Pomega(\Aut,\check{f})\cap \Pomega(\Aut,\FcS_P)\neq\emptyset$.\\
  \item Show \eqref{equ:Pf:a}: Now we pick any $\pi\in\Pomega(\Aut,\check{f})\cap \Pomega(\Aut,\FcS_P)$ and show that $\pi\in\Pomega(\Aut,\FcR_S)$. First, as $M$ is deterministic and has distinct transition labels, we have $|\ON{Paths}^{-1}_{\Aut}(\pi)|=1$ and define $s:=\ON{Paths}^{-1}_{\Aut}(\pi)$. As $\pi\in\Pomega(\Aut,\FcS_P)$ we have $s\in\Lomega(\Aut,\FcS_P)$. As $\Lomega_P=\Lomega(\Aut,\FcS_P)$  (from \REFprop{prop:StreetRabinAut} (b)) we have $s\in\Lomega_P$. 
  Further, as $\pi\in\Pomega(\Aut,\check{f})$ we know that for any $\q,\q'\in \Q$ and $\nu\in\Q^*$ s.t.\ $\nu \q \q'\in\LPrefix{\pi}$, there exists an event $\sigma\in \check{f}(\nu \q)$ s.t.\ $\q'=\delta(\q,\sigma)$. Now it follows that indeed $s'=\ON{Paths}^{-1}_{\Aut}(\nu\q)$ is uniquely defined and $s'\sigma\in \LPrefix{s}$ by definition. Now it follows from  \eqref{equ:fvscheckf} that $\sigma\in f(s')$ which implies $s\in\Lomega_f$. I.e., we have $s\in\Lomega_P\cap\Lomega_f$. 
  As \eqref{equ:f:contain} holds we have $s\in\Lomega_P$ and $s\in\Lomega_S$. Now recall from \REFprop{prop:StreetRabinAut} (a) that $\Lomega_P\subseteq \Lomega(M)$ and therefore $s\in\Lomega(M)\cap\Lomega_S$.
 Now it follows from \REFprop{prop:StreetRabinAut} (c) that $\Lomega_S\cap\Lomega(M)=\Lomega(\Aut,\FcR_S)$. This implies $\pi \in \Pomega(\Aut,\FcR_S)$.\\
 \item Show \eqref{equ:Pf:b}: Pick $\nu\in\Pstar(\Aut,\check{f})$ and observe that this implies $\nu\in\Pstar(\Aut)$. As $\Lstar_P=\Lstar(\Aut)$ (from \REFprop{prop:StreetRabinAut} (a)) this implies that $t=\ON{Paths}^{-1}_{\Aut}(\nu)\in\Lstar_P$. Further, it follows from the same reasoning as before that $t\in\Lstar_f$. As \eqref{equ:f:nonconf} holds, this implies $t\in\LPrefix{\Lomega_f\cap\Lomega_P}$ and thereby $\nu\in\LPrefix{\Pomega(\Aut,\check{f})\cap\Pomega(\Aut,\FcS_P)}$.
\end{inparaitem}\\
 \item \enquote{$\Leftarrow$}:
  Fix $\check{f}$ s.t.\ \eqref{equ:Pf} holds, and $f$ s.t. \eqref{equ:fvscheckf} holds. \\
 \begin{inparaitem}[$\triangleright$]
 \item We first show that $\Lomega_f\cap\Lomega_P\neq\emptyset$.
 To this end, recall that \eqref{equ:Pf:a} holds, i.e., $\emptyset\subsetneq\Pomega(\Aut,\check{f})\cap \Pomega(\Aut,\FcS_P)$. Hence, there exists a path $\pi\in\Pomega(\Aut,\FcS_P)$ s.t.\ also $\pi\in\Pomega(\Aut,\check{f})$.
 Now it follows form the same reasoning as in the proof of \eqref{equ:Pf:a} (in \enquote{$\Rightarrow$} above) that for $s:=\ON{Paths}^{-1}_{\Aut}(\pi)$ holds $s\in\Lomega_P\cap\Lomega_f\neq\emptyset$.\\
 \item Show \eqref{equ:f:contain}: Now pick any $s\in\Lomega_P\cap\Lomega_f$ and show $s\in\Lomega_S$. It again follows from the same reasoning as before (see first item in \enquote{$\Rightarrow$} above) that for $\pi:=\ON{Paths}_{\Aut}(s)$ holds $\pi\in\Pomega(\Aut,\check{f})\cap \Pomega(\Aut,\FcS_P)$. As \eqref{equ:Pf:a} holds, this implies $\pi\in\Pomega(\Aut,\FcR_S)$ and therefore $s\in\Lomega(\Aut,\FcR_S)$. Then it follows from condition (c) in \REFprop{prop:StreetRabinAut} that $s\in\Lomega_S$.\\ 
 \item Show \eqref{equ:f:nonconf}:
 Pick $t\in\Lstar_f\cap\Lstar_P$ and define $\nu=\ON{Paths}^{-1}(t)$. Then it follows from the same reasoning as before that $t\in\Lstar_f$ implies $\nu\in\Pstar(\Aut,\check{f})$. As \eqref{equ:Pf:b} holds, this implies $\nu\in\LPrefix{\Pomega(\Aut,\check{f})\cap\Pomega(\Aut,\FcS_P)}$
 and thereby $t\in\LPrefix{\Lomega_f\cap\Lomega_P}$.
 \end{inparaitem}
\end{inparaitem}
\end{proof}

Algorithms solving various versions of \REFproblem{prob:SCT_P} are studied by Thistle and Wonham in \cite{TW1991,TW1994a,TW1994b,thistle1995}. 
All of them are initialized with a deterministic finite-state machine $\Aut$ equipped with two acceptance conditions $\Fc_P$ and $\Fc_S$
where $\Fc_S$ is a Rabin condition.
However, $\Fc_P$ is chosen to be trivial in \cite{TW1991} (i.e., $\Lomega_P = \ASigma^\omega$), a deterministic Büchi condition in \cite{TW1994a,TW1994b}, and a deterministic Streett condition in \cite{thistle1995}. 
I.e., the algorithm in \cite{thistle1995} solves \REFproblem{prob:SCT_P} but does not allow for a direct symbolic implementation.
In the remaining sections of this paper, we show an alternative way to solve \REFproblem{prob:SCT_P} which establishes a new connection between the fields of supervisory control and reactive synthesis. This allows to utilize symbolic algorithms from reactive synthesis to solve \REFproblem{prob:SCT_P}.

\begin{remark}
We remark that \REFthm{thm:SCTvsSCTP} similarly holds for the $*$-language case. Here, the DFA $(M,F)$ captures the synthesis problem for the $*$-language tuples $(\Lstar_P, \Lomega_P)$ and  $(\Lstar_S, \Lomega_S)$ as discussed in  \REFsec{sec:real}. In the classical supervisor synthesis procedure (see e.g.\ \cite{Cassandras}), this synthesis automaton $(M,F)$ is manipulated in various steps to compute a supervisor $f$ solving \REFproblem{prob:SCT}. 
As this synthesis procedure takes an automaton as input and computes a path-based supervisor as an output, it indeed solves \REFproblem{prob:SCT_P} (with input $(M,F)$). 
Hence, this algorithm only provides a solution to \REFproblem{prob:SCT} if  \REFproblem{prob:SCT} and \REFproblem{prob:SCT_P} are indeed equivalent for every regular input, which is known to be true (see e.g., \cite{Cassandras,DBLP:conf/ifac/SchmuckMS20} for details).

We want to point out that solving \REFproblem{prob:SCT_P} with $\omega$-regular parameters is much more difficult then solving the (classical) $*$-language version. This is due to the fact that the $\omega$-regular winning conditions of the plant and the specification (captured by $\FcS_P$ and $\FcR_S$, respectively), do not need to be fulfilled simultaneously. 
\end{remark}

\subsection{Another Simple Example}\label{sec:example}
Consider the finite state machine $\Aut$ depicted in \REFfig{fig:M} for a path-based supervisor synthesis problem. Here, the alphabet is $\ASigma = \set{a, b, c}$, partitioned in $\ASigmac =\set{a,b}$ (indicated by a tick on the
corresponding edges in \REFfig{fig:M}) and $\ASigmauc = \set{c}$.
The \emph{uncontrolled} plant is assumed to only generate traces allowed in $\Aut$ (safety) and to additionally visit the state $p_2$ always again (liveness).
The latter is modelled by a B\"uchi acceptance condition
$\FcB_P= \set{p_2}$, and indicated
by the light-blue double circle around state $p_2$ in \REFfig{fig:M}. The Büchi condition $\FcB_P$ can be equivalently formulated as the Streett condition $\FcS_P=\set{\tuple{\set{p_0,p_1,p_2},\set{p_2}}}$.

The specification requires that the \emph{controlled} plant should visit  
state $p_1$ always again (liveness) and does not dead-lock (safety). 
This can be modelled by a B\"uchi condition with $\FcB_S = \set{p_1}$, indicated by the red
double circle around $p_1$ in \REFfig{fig:M}. Again, we can equivalently represent $\FcB_S$ as the Rabin condition $\FcR_S=\set{\tuple{\set{p_1},\emptyset}}$. 
In order to achieve the desired specification, the supervisor can only disable controllable actions; thus, every control pattern allows $c$.

The supervisor synthesis problem, \REFproblem{prob:SCT_P}, now asks to synthesize a path-based supervisor that ensures  
(i) whenever $p_2$ is always visited again, also $p_1$ is always visited again,
and that 
(ii) the controller never prevents the plant from visiting $p_2$ again in the future.
A path-based supervisor solving this problem is given by the following rule: any path ending in $p_0$ is mapped to $\set{a,c}$, any path ending in $p_1$ is mapped to $\set{b, c}$,
and any path ending in $p_2$ is mapped to $\set{a, c}$. This effectively disables the self-loop on event $b$ in state $p_2$. 

Note that this solution is not unique: for each $n\geq 0$, a supervisor could map paths ending in $p_0$ and $p_1$
as before, but map paths ending in $p_2$ to $\set{a,b,c}$ if the number of visits to $p_2$ is less than $n$
and to $\set{a,c}$ otherwise. In this case, the supervisor would allow the self loop on $b$ in $p_2$ to be taken $n$ number of times. A deterministic supervisor must however decide on a fixed $n$ after which $b$ is disabled in $p_2$, as otherwise the specification might not be fulfilled on a path fulfilling the plants' liveness assumption.

\begin{remark}
 The above example demonstrates the well-known fact that for $\omega$-languages, there may not exist a 
\emph{maximally permissive} supervisor, i.e., a supervisor $f$ solving \REFproblem{prob:SCT} s.t.\ $\Lomega_{f'}\cap\Lomega_P\subseteq\Lomega_f\cap\Lomega_P$ for all other supervisors $f'$ solving \REFproblem{prob:SCT}.
In the above example, the maximal permissive supervisor would need to \enquote{eventually} disable $b$ which cannot be modeled by a supervisor mapping from \emph{finite} past strings to control patters. 
The induced language $\Lomega_f$ of $f$ is always topologically closed, disallowing the introduction of new liveness properties.
This situation is in contrast to supervisor synthesis for $*$-languages where maximally permissive solutions to \REFproblem{prob:SCT} always exist.
\end{remark}

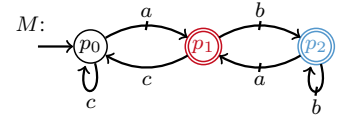
\begin{figure}[t]
\begin{center}
 \begin{tikzpicture}[auto,scale=1]
  \begin{footnotesize}
       \node (name) at (-0.8,0.3) {$M$:};
          \node (init) at (-0.8,0) {};
          \node[mystate] (p0) at (0,0) {$p_0$};
          \node[mystate,accepting,sws-red] (p1) at (1.5,0) {$p_1$};
          \node[mystate,accepting,sws-blue] (p2) at (3,0) {$p_2$};
          
          
\SFSAutomatEdge{init}{}{p0}{}{}  
\SFSAutomatEdgeCrossed{p0}{a}{p1}{bend left}{}
\SFSAutomatEdge{p0}{c}{p0}{loop below}{}
\SFSAutomatEdge{p1}{c}{p0}{bend left}{}
\SFSAutomatEdgeCrossed{p1}{b}{p2}{bend left}{}
\SFSAutomatEdgeCrossed{p2}{b}{p2}{loop below}{}
\SFSAutomatEdgeCrossed{p2}{a}{p1}{bend left}{}

  \end{footnotesize}
\end{tikzpicture}
 \vspace{-0.6cm}
 \end{center}
 \caption{Finite state machine $\Aut$ for example in \REFsec{sec:example}. Plant and specification markings indicated in blue ($p_2$) and red ($p_1$), respectively. Controllable events $\set{a,b}$ indicated by a ticked transition. Event $c$ is uncontrollable.}\label{fig:M}
  \vspace{-0.2cm}
\end{figure}

\section{From Supervisor Synthesis to Games}\label{sec:games}

We shall reduce \REFproblem{prob:SCT_P} to solving a class of two-player games
on graphs with $\omega$-regular winning conditions.

\subsection{Two-Player Games}\label{sec:OG:prelim}

A \emph{two-player game graph} $\GG = \Tuple{\Qz,\Qo, \Trz,\Tro,q_{init}}$ consists of two finite disjoint state sets 
$\Qz$ and $\Qo$, two transition functions $\Trz: \Qz \rightarrow \twoup{\Qo}$ and $\Tro: \Qo \rightarrow \twoup{\Qz}$,
and an initial state $q_{init}\in\Qz$. We write $Q = \Qz\cup\Qo$.
Given a game graph $\GG$, a \emph{strategy} for player $0$ is a function $\fz{}: q_{init}(\Qo\Qz)^*\rightarrow \Qo$. 
The sequence $\rho\in Q^\infty$ is called a play over $\GG$ if $\rho(0)=q_{init}$ and for all $k\in\ON{Length}(\rho)-1$,
we have $\rho(k+1)\in\Trz(\rho(k))$ if $\ON{Last}(\rho)\in\Qz$ and $\rho(k+1)\in\Tro(\rho(k))$ otherwise. 
The play $\rho$ is \emph{compliant} with $\fz{}$ if additionally 
$\fz{}(\rho|_{[0,k]}) = \rho(k+1)$ if $\ON{Last}(\rho)\in\Qz$. 
We denote by $\Pstar(\GG,\fz{})$ and $\Pomega(\GG,\fz{})$ 
the set of finite and infinite plays over $\GG$ compliant with $\fz{}$.

We define $\omega$-regular winning conditions for two-player games.
These are specified analogously to acceptance conditions for finite state machines over subsets of states $Q$. 
That is, we consider Büchi, Rabin and Streett conditions $\Fc$ as defined in \REFsec{sec:SCT:prelim} over subsets of $Q$
and say that a play $\rho$ is winning w.r.t.\ $\Fc$ if $\rho$ satisfies $\Fc$ on $\GG$. 
In addition, we also consider the \emph{parity} accepting condition \cite{EmersonJutla_1991}.
For the parity condition with $k$ parities, we assume there is a coloring function $\Omega: Q \rightarrow \set{0,\ldots, k-1}$.
A play $\rho$ is winning if the maximum color seen infinitely often is even.

We call a game graph equipped with a Büchi, Rabin, Streett, or parity winning condition $\Fc$ a Büchi, Rabin, Streett, or parity game, respectively, 
and denote it by the tuple $(\GG,\Fc)$. 
The set of all winning plays over $\GG$ w.r.t.\ $\Fc$ is denoted $\Pomega(\GG,\Fc)$.
A strategy $\fz{}$ is \emph{winning} in a game $(\GG,\Fc)$, if $\Pomega(\GG,\fz{})\subseteq\Pomega(\GG,\Fc)$. We remark that it is decidable if player~0 has a winning
strategy in a two-player game with a B\"uchi, Rabin, Streett, or parity winning
condition
\cite{PnueliRosner89,EmersonJutla88,EmersonJutla_1991,Thomas1995}.


\subsection{Supervisor Synthesis as a Two-Player Game}\label{sec:compare:plays}

Intuitively, one can interpret the interaction of a supervisor with the plant as a two-player game over $\Aut$. 
Player $0$ (the supervisor) picks a control pattern $\gamma\in\Gamma$ and player $1$ (the plant) resolves the remaining 
non-determinism by choosing a transition allowed by $\gamma$.
We formalize the construction below. 

\begin{definition}\label{def:MtoG}
 Let $\Aut=(\Q,\,\Sigma,\,\delta,\, q_0)$ be as in \REFprop{prop:StreetRabinAut} with $\Sigma_{uc}\subseteq\Sigma$ and 
$\Gamma:=\{\,\gamma\subseteq\ASigma\,|\,\ASigmauc\subseteq\gamma\,\}$. 
Then we define its associated game graph as $\GG(\Aut) = \Tuple{\Qz,\Qo, \Trz,\Tro,q_0}$ s.t.\
 \begin{compactitem}
  \item $\Qz=\Q$
  \item $\Qo=\Q\times\Gamma$
  \item $\Trz(\q)=\Set{\q}\times\Gamma$
  \item $\q'\in\Tro((\q,\gamma))$ iff $\sigma\in\gamma$ and $\q'=\delta(\q,\sigma)$. 
 \end{compactitem}
\end{definition}

Intuitively, the game graph $\GG$ makes the choice of the control pattern taken by the state-based supervisor over $\Aut$ explicit by 
inserting player $1$ states in between any two player $0$ states. 
I.e, the choice of control pattern $\gamma$ in state $\q\in\Q$ of $\Aut$ corresponds to the move of player $0$ from $q=\q$ to $q'=(\q,\gamma)$ in $\GG$. 
Further, as $\Aut$ is assumed to have unique transition labels, this expansion allows to remove all transition labels resulting in 
an unlabeled game graph $\GG$ as defined in \REFsec{sec:OG:prelim}.
\REFfig{fig:MtoG} shows the two-player game graph $\GG(\Aut)$ corresponding to $\Aut$ in \REFfig{fig:M}.

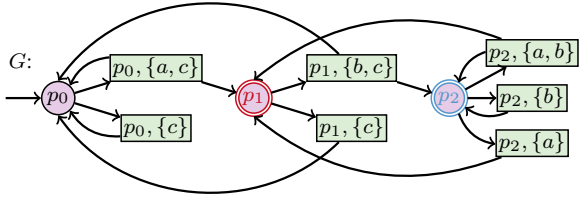
\begin{figure}[t]
\begin{center}
  \begin{tikzpicture}[auto,scale=1]
  \begin{footnotesize}
    \node (name) at (-0.5,0.5) {$G$:};
          \node (init) at (-0.8,0) {};
          \node[mystate,fill=violet!20] (p0) at (0,0) {$p_0$};
          \node[mysquare,draw,fill=\lgreen] (p0c) at (1.3,-0.4) {$p_0, \set{c}$};
          \node[mysquare,draw,fill=\lgreen] (p0ac) at (1.3,0.4) {$p_0, \set{a,c}$};
          \node[mystate,accepting,sws-red,fill=violet!20] (p1) at (2.6,0) {$p_1$};
          \node[mysquare,draw,fill=\lgreen] (p1bc) at (3.9,0.4) {$p_1, \set{b,c}$};
          \node[mysquare,draw,fill=\lgreen] (p1c) at (3.9,-0.4) {$p_1, \set{c}$};
          \node[mystate,accepting,sws-blue,fill=violet!20] (p2) at (5.2,0) {$p_2$};
          \node[mysquare,draw,fill=\lgreen] (p2ab) at (6.3,0.6) {$p_2, \set{a,b}$};
          \node[mysquare,draw,fill=\lgreen] (p2b) at (6.3,0) {$p_2, \set{b}$};
          \node[mysquare,draw,fill=\lgreen] (p2a) at (6.3,-0.6) {$p_2, \set{a}$};
          
          
          
\SFSAutomatEdge{init}{}{p0}{}{}  
\SFSAutomatEdge{p0}{}{p0c}{}{}
\SFSAutomatEdge{p0}{}{p0ac}{}{}
\SFSAutomatEdge{p0ac}{}{p0}{bend right}{}
\SFSAutomatEdge{p0ac}{}{p1}{}{}
\SFSAutomatEdge{p0c}{}{p0}{bend left}{}
\SFSAutomatEdge{p1}{}{p1bc}{}{}
\SFSAutomatEdge{p1}{}{p1c}{}{}
\SFSAutomatEdge{p1bc}{}{p2}{}{}
\SFSAutomatEdge{p1bc}{}{p0.north}{bend right=50}{}
\SFSAutomatEdge{p1c}{}{p0.south}{bend left=50}{}
\SFSAutomatEdge{p2}{}{p2ab}{}{}
\SFSAutomatEdge{p2}{}{p2a}{bend right}{}
\SFSAutomatEdge{p2}{}{p2b}{}{}
\SFSAutomatEdge{p2ab}{}{p2}{bend right}{}
\SFSAutomatEdge{p2b}{}{p2}{bend left}{}
\SFSAutomatEdge{p2ab}{}{p1.north}{bend right}{}
\SFSAutomatEdge{p2a}{}{p1.south}{bend left}{}

  \end{footnotesize}
\end{tikzpicture}
\end{center}
 \vspace{-0.8cm}
 \caption{Game graph $\GG(\Aut)$ associated with the synthesis automaton $\Aut$ in \REFfig{fig:M}. Supervisor and plant player states are indicated by 
a circular violet and a rectangular green shape, respectively. 
Rectangular states $(p,\gamma)$ indicate the control pattern\protect\footnotemark $\gamma$ chosen by the supervisor in state $p$ of $\Aut$.}\label{fig:MtoG}
 \vspace{-0.4cm}
\end{figure}
\footnotetext{We restrict depicted control patterns to events enabled at the source state.}

We now discuss an appropriate winning condition for the game.
Consider the state-based supervisor synthesis problem (\REFproblem{prob:SCT_P}) over the Streett/Rabin supervisor synthesis automaton $(\Aut,\FcS_P,\FcR_S)$. 
Here, \eqref{equ:Pf:a} requires that any infinite trace over $\Aut$ which is both compliant with $f$ and fulfills the plant assumption 
$\Lomega_P$ also fulfills the specification $\Lomega_S$. Hence, we can equivalently write \eqref{equ:Pf:a} as the implication
\begin{equation}\label{equ:Pf:beta}
 \AllQ{\pi\in\Pomega(\Aut,\check{f})}{\propImp*{\pi\in \Pomega(\Aut,\FcS_P)}{\pi\in\Pomega(\Aut,\FcR_S)}},
\end{equation}
which is in turn equivalent to 
\begin{equation}\label{equ:Pf:gamma}
 \AllQ{\pi\in\Pomega(\Aut,\check{f})}{\propDisj*{\pi\notin \Pomega(\Aut,\FcS_P)}{\pi\in\Pomega(\Aut,\FcR_S)}}.
\end{equation}
Consequently, \eqref{equ:Pf:a} is achieved by a supervisor $\check{f}$ which ensures plays over $\Aut$ 
either
\emph{do not} satisfy the Streett condition $\FcS_P$ \emph{or} 
fulfill the Rabin condition $\FcR_S$. 
However, as Rabin and Streett conditions are duals, \emph{not satisfying} the 
\emph{Streett} condition $\FcS_P$ is equivalent to 
\emph{satisfying} the \emph{Rabin} condition $\FcR_P:=\neg      \FcS_P$. 
Further, given the definition of Rabin winning conditions (see \REFsec{sec:SCT:prelim}), it is easy to see that a path 
over $\Aut$ satisfies either the Rabin condition $\FcR_P$ \emph{or} the Rabin condition $\FcR_S$ iff 
it satisfies the Rabin condition $\FcR_{P\rightarrow S}=\FcR_P\cup\FcR_S$. 
With this observation, we can further rewrite \eqref{equ:Pf:a} into the equivalent formula
\begin{equation}\label{equ:Pf:a:new}
 \Pomega(\Aut,\check{f})\subseteq \Pomega(\Aut,\FcR_{P\rightarrow S}).
\end{equation}
Thus, an obvious choice for the winning condition over the game graph $\GG(\Aut)$ is the Rabin condition $\FcR_{P\rightarrow S}$.

\begin{example}\label{expl:rabin}
 Consider the example from \REFsec{sec:example} and recall that $\FcS_P=\set{\tuple{\set{p_0,p_1,p_2},\set{p_2}}}$ and $\FcR_S=\set{\tuple{\set{p_1},\emptyset}}$. This gives the Rabin winning condition
\begin{equation}\label{ex:rabin}
\FcR_{P\rightarrow S}=\set{\tuple{\set{p_0,p_1,p_2},\set{p_2}}, \tuple{\set{p_1}, \emptyset}}
\end{equation}
for the induced game over $\GG(\Aut)$. Intuitively, the condition in \eqref{ex:rabin} states that either $p_2$  is only visited \emph{finitely often} (first Rabin pair) or $p_1$ is visited \emph{infinitely often} (second Rabin pair). These two possibilities admit winning strategies that either prevent the plant from fulfilling its liveness properties (e.g.\ by always disabling $a$ and $b$ in all states) or that ensure that the specification gets fulfilled (e.g.\ by choosing the strategy given in \REFsec{sec:example}).
\end{example}

As the above example demonstrates, a winning strategy for $\FcR_{P\rightarrow S}$ may not fulfill condition \eqref{equ:Pf:b}. A strategy can choose to satisfy  \eqref{ex:rabin} vacuously, by actively preventing the plant to fulfill its liveness properties.
Thus, we need to modify the winning condition to ensure the resulting strategy satisfies
both \eqref{equ:Pf:a} and \eqref{equ:Pf:b}. As the non-conflicting requirement of \eqref{equ:Pf:b} is not a linear property \cite{EhlersJdeds}, it cannot be easily \enquote{compiled away} in reactive synthesis.
Therefore, we consider a different type of game instead, called \emph{obliging game}.

\section{Supervisor Synthesis via Obliging Games}\label{sec:compare}

\subsection{Obliging Games}

An obliging game \cite{chatterjee2010obliging} 
is a triple $(\GG,\FcStrong,\FcWeak)$ where $\GG$ is a game graph and $\FcStrong$ and $\FcWeak$ are two
winning conditions, called strong and weak, respectively. 
To win an obliging game, player $0$ (the \enquote{controller}) needs to \emph{ensure} the strong winning condition $\FcStrong$ against any 
strategy of player $1$ (the \enquote{system}), while \emph{allowing} the system to cooperate with him 
to \emph{additionally} fulfill $\FcWeak$. 
Such winning strategies are therefore called \emph{gracious} and the synthesis problem for obliging games asks 
to synthesize such a gracious control strategy or determine that none exists, as formalized in the following problem statement.

\begin{problem}[Obliging Games]\label{prob:OG}
Given an obliging game $(\GG,\FcStrong,\FcWeak)$, synthesize a strategy $\fz{}$ for player $0$ s.t.\
\begin{subequations}\label{equ:fz}
 \begin{compactenum}[(i)]
 \item every play over $\GG$ compliant with $\fz{}$ is winning w.r.t.\ $\FcStrong$, 
 \begin{equation}\label{equ:fz:strong}
  \Pomega(\GG,\fz{})\subseteq \Pomega(\GG,\FcStrong)
 \end{equation}
 \item for every finite play $\nu$ over $\GG$ compliant with $\fz{}$, 
  there exists an infinite play $\rho$ over $\GG$ compliant with $\fz{}$ and winning 
  w.r.t.\ $\FcWeak$, s.t.\ $\nu\in\OpPre{\rho}$, i.e.,
 \begin{equation}\label{equ:fz:weak}
  \Pstar(\GG,\fz{}) \subseteq \LPrefix{\Pomega(\GG,\fz{})\cap\Pomega(\GG,\FcWeak)},
 \end{equation}
\end{compactenum}
\end{subequations}
or determine that no such strategy exists.
\end{problem}

The following theorem characterizes the solution of \REFproblem{prob:OG} by a reduction to a parity game. As parity games are decidable and one can effectivly construct winning strategies of player $0$ in such games, \REFthm{thm:obliging-games} establishes that the same is true for obliging games. 

\begin{theorem}
\label{thm:obliging-games}
Every obliging game $(G, \FcStrong, \FcWeak)$ is reducible
to a two-player game with an $\omega$-regular winning condition. 
In particular, an obliging game $(G, \FcR, \FcS)$ with $n$ states, a Rabin condition $\FcR$ with $k$
pairs, and a Streett condition $\FcS$ with $l$ pairs can be reduced to a two-player game with $nk^2k!2^{O(l)}$
states, a parity condition with $2k + 2$ colors, and $2k2^{O(l)}$ memory.
\end{theorem}

\begin{proof}
 The first claim follows from \cite[Thm.3]{chatterjee2010obliging}. 
For the second claim, recall that from a Streett condition with $l$ pairs, one can construct a (non-deterministic) B\"uchi automaton with $2^{O(l)}$ states that accepts the same language.
Moreover, by taking a product with a monitor with $k^2\cdot k!$ states, we can convert the Rabin condition to a parity condition \cite{EmersonJutla_1991} with $2k$ colors.
Now, the construction in \cite[Lem.2, Thm.4]{chatterjee2010obliging} reduces an obliging game with $\tilde{n}$ states, a strong parity winning condition with $2k$ colors and a weak winning condition accepted by a B\"uchi automaton with $q$ states into a
game with $O(\tilde{n} q)$ states, $2k+2$ colors, and memory $2qk$. 
Applying this reduction to our setting yields a parity game with $n\cdot k^2\cdot k!\cdot 2^{O(l)}$ states, 
$2k+2$ colors, and memory $2k\cdot 2^{O(l)}$.
\end{proof}
In order to reduce the supervisor synthesis problem to obliging games we need to define appropriate winning conditions.
We can see by inspection that after replacing \eqref{equ:Pf:a} by \eqref{equ:Pf:a:new} in \REFproblem{prob:SCT_P} and 
defining $\FcStrong:=\FcR_{P\rightarrow S}$ and $\FcWeak:=\FcS_P$ in \REFproblem{prob:OG}, 
the two problem descriptions match. 
However, the system models and the corresponding control mechanisms are different. 
We therefore need to match path-based supervisors for $\Aut$ with player~$0$ strategies over $\GG(\Aut)$. 

\subsection{Formal Reduction}

Given the reduction from $\Aut$ to a game graph $\GG(\Aut)$, and the strong and weak winning conditions,
it remains to show that the resulting obliging game is indeed equivalent to the path-based supervisor synthesis problem.
This is formalized in the following theorem.

\begin{theorem}\label{thm:SCTPvsOG}
Let $(\Aut,\FcS_P,\FcR_S)$ be a Streett/Rabin supervisor synthesis automaton 
and $\GG(\Aut)$ its associated game graph. 
Then there exists a path-based supervisor $\check{f}$ that is a solution to the supervisor synthesis problem over $(\Aut,\FcS_P,\FcR_S)$ iff 
there exists a player $0$ strategy $h$ winning the obliging game $(\GG(\Aut), \FcR_{P\rightarrow S}, \FcS_P)$.
\end{theorem}

In order to prove \REFthm{thm:SCTPvsOG} we first formalize a mapping from paths over $\Aut$ to plays over $\GG$ and back. 
This will allow us to define corresponding path-based supervisors and gracious strategies and formalize their associated 
properties. 

\smallskip
\noindent\textit{Paths vs. Plays.}
To formally connect paths in $\Aut$ to plays over $\GG$, we define the set-valued map 
$\ON{Plays}: \q_0\Q^*\fun 2^{q_0(\Qo\Qz)^*}$ iteratively 
as follows:
 $\ON{Plays}(\q_0):=\set{\q_0}$ and $\ON{Plays}(\nu\q):=\SetComp{\mu \tilde{\q}\q}{\mu\in\ON{Plays}(\nu), \tilde{\q}\in\set{\ON{Last}(\nu)}\times\Gamma}$. 
By slightly abusing notation, we extend the map $\ON{Plays}$ to infinite paths 
$\pi\in\q_0\Q^\omega$ as the limit of all mappings $\ON{Plays}(p_n)$ where $(p_n)\in \q_0\Q^*$ is the unbounded monotone sequence of prefixes of $\pi$.
Similarly, we define the inverse map $\ON{Plays}^{-1}: q_0(\Qo\Qz)^*\fun \q_0\Q^*$ s.t.\ $\ON{Plays}^{-1}(\mu)=\nu$ 
where $\nu$ is the single element of the set $\SetComp{\nu\in \q_0\Q^*}{\mu\in\ON{Plays}(\nu)}$. Again we extend $\ON{Plays}^{-1}$ to infinite strings in the obvious way.

The construction of $\GG(\Aut)$ from $\Aut$ in \REFdef{def:MtoG} allows us to show that the map $\ON{Plays}$ indeed captures all 
the information required to map paths over $\Aut$ to the corresponding plays over $\GG(\Aut)$ and vice versa. 

\begin{lemma}\label{lem:MvsG}
Let $\Aut$ be a finite state machine as in \REFprop{prop:StreetRabinAut} and $\GG(\Aut)$ its associated game graph as in \REFdef{def:MtoG}. Then
\begin{subequations}\label{equ:plays}
\begin{align}
 \ON{Plays}(\Pstar(\Aut))&=\Pstar(\GG),\label{equ:plays:a}\\
 \ON{Plays}(\Pomega(\Aut))&=\Pomega(\GG),~\text{and}\label{equ:plays:b}\\
 \ON{Plays}(\Pomega(\Aut,\Fc))&=\Pomega(\GG,\Fc),\label{equ:plays:c}
\end{align}
where $\mathcal{F}$ is a winning condition over $\Aut$.
\end{subequations}
\end{lemma}

\begin{proof}
\begin{inparaitem}[$\blacktriangleright$]
 \item \eqref{equ:plays:a}: Let $\nu=x_0x_1\hdots x_k\in\Pstar(\Aut)$. Then $\ON{Plays}(\nu)$ is the set containing all plays $\mu:=x_0(x_0,\gamma_0)x_1(x_1,\gamma_1)\hdots(x_k,\gamma_{k-1})x_k$ s.t.\ $\gamma_i\in\Gamma$ for all $i\in[0;k]$. It follows from \REFdef{def:MtoG} that all $\mu\in\ON{Plays}(\nu)$ are a play over $\GG$ starting in $q_0$, and, hence $\ON{Plays}(\Pstar(\Aut))\subseteq\Pstar(\GG)$. The inverse direction follows similarly from  the last condition in \REFdef{def:MtoG}.
  \item \eqref{equ:plays:b} follows directly from \eqref{equ:plays:a} by taking the limit closure on both sides.
 \item \eqref{equ:plays:b} First, any winning condition over $\Aut$ is also a winning condtion over $\GG$ as $Q=\Qz\cup\Qo$ with $\Qz=\Q$.Now pick any path $\pi\in\Pomega(\Aut,\Fc)$. Then we know that the set $\ON{Inf}(\pi)\subseteq\Q$ fulfills the conditions for acceptance w.r.t.\ the acceptance condition $\Fc$ over $\Aut$. Now take any $\rho\in \ON{Plays}(\pi)\subseteq\Pomega(\GG)$ and observe that deciding winning of $\rho$ w.r.t.\ $\Fc\subseteq\Qz$ only depends on the set $\ON{Inf}(\rho)|_{\Qz}\subseteq\Qz$. Then the claim follows from the observation that the definition of $\ON{Plays}$ implies $\ON{Inf}(\rho)|_{\Qz}=\ON{Inf}(\pi)$. 
\end{inparaitem}
\end{proof}

\smallskip
\noindent\textit{Supervisors vs. Strategies.}
Unfortunately, we cannot directly utilize the properties in \eqref{equ:plays} to relate 
path-based supervisors and gracious strategies. 
By definition, control strategies can base their decision on all information from the past observed state sequence. 
As one path over $\Aut$ corresponds to multiple plays over $\GG(\Aut)$, every 
such play could in principle induce a different control decision. 
We call strategies that do not utilize this additional flexibility  \emph{non-ambiguous}.

\begin{definition}\label{def:nonambiguous}
Let $\GG$ be as in \REFdef{def:MtoG}. 
We call a player $0$ strategy over $\GG$ \emph{non-ambiguous} if for any $\nu\in\q_0\Q^*$ and any $\mu,\mu'\in \ON{Plays}(\pi)$,
we have $\check{h}(\mu)=\check{h}(\mu')$. 
\end{definition}

A strategy over $\GG$ can only choose one particular next state in a current one. 
As the initial state is unique, there must be a unique control pattern chosen in this state leading to a unique next state in $\GG$. 
Iteratively applying this argument shows that there is a unique play over $\GG$ generated under any control strategy $h$. 
Therefore, we can always construct a non-ambiguous strategy $\check{h}$ over $\GG$ from a given control strategy $h$ with the same set of generated plays.

\begin{proposition}\label{prop:nonambiguous}
 Given the premises of \REFlem{def:nonambiguous}, let $\fz{}$ be a strategy over $\GG$, then $\check{h}$ s.t.\ 
 \begin{equation}\label{equ:checkh}
  \check{h}(\q_0):=\fz{\q_0},\quad\check{h}(\mu \tilde{\q}_k \q_{k+1}):=\fz{\mu \check{h}(\mu) \q_{k+1}}
 \end{equation}
  is a non-ambiguous  player $0$ strategy over $\GG$ and it holds that $\Pomega(\GG,h)=\Pomega(\GG,\check{h})$.
\end{proposition}

\begin{proof}
 For the base-case, $\ON{Plays}(\q_0)=\set{\q_0}$ and therefore for all $\mu,\mu'\in \ON{Plays}(\q_0)$ we have $\mu=\mu'=\q_0$.  Hence $\check{h}(\mu)=\check{h}(\mu')=\fz{\q_0}$, i.e., $\Pomega(\GG,h)|_{[0,0]}=\Pomega(\GG,\check{h})|_{[0,0]}$.
 For the induction step, fix $\nu\in\q_0\Q^*$ with $|\nu|=k>1$ and assume that for all $\mu,\mu'\in \ON{Plays}(\nu)$ we have $\check{h}(\mu)=\check{h}(\mu')$. Now choose any $\q\in\Q$ and observe that $\ON{Plays}(\nu \q)=\SetComp{\mu\tilde{\q}\q}{\mu\in\ON{Plays}(\nu), \tilde{\q}\in\set{\ON{Last}(\nu)}\times\Gamma}$. Now pick any two $\mu\tilde{\q}\q, \mu'\tilde{\q}'\q\in\ON{Plays}(\nu \q)$ and observe that from the definition of $\check{h}$ follows that $\check{h}(\mu\tilde{\q}\q)=\fz{\mu \check{h}(\mu) \q}$ and $\check{h}(\mu'\tilde{\q}'\q)=\fz{\mu' \check{h}(\mu') \q}$. As $\mu,\mu'\in\ON{Plays}(\nu)$ it follows from the induction hypothesis that $\check{h}(\mu)=\check{h}(\mu')$ and therefore $\check{h}(\mu\tilde{\q}\q)=\check{h}(\mu'\tilde{\q}'\q)$. This proves that $\check{h}$ is non-ambiguous.
  Now assume $\Lambda:=\Pomega(\GG,h)|_{[0,k]}=\Pomega(\GG,\check{h})|_{[0,k]}$ and $\check{h}(\mu)=\fz{\mu}$ for all $\mu\in\Lambda$. Then it follows from \REFdef{def:MtoG} that $\Pomega(\GG,h')|_{[0,k+2]}$ contains all strings $\mu(\ON{Last}(\mu),\gamma)\q'$ s.t.\ $\mu\in\Lambda$, $(\ON{Last}(\mu),\gamma)=h'(\mu)$ and $\q'=\delta(\q_0,\sigma)$ for some $\sigma\in\gamma$. With this it immediately follows from the induction hypothesis that $\Pomega(\GG,h)|_{[0,k+2]}=\Pomega(\GG,\check{h})|_{[0,k+2]}$. As both $\Pomega(\GG,h)$ and $\Pomega(\GG,\check{h})$ are closed languages, this proves the claim.
\end{proof}

\REFprop{prop:nonambiguous} shows that restricting attention to non-ambiguous player $0$ strategies over $\GG$ is without loss of generality. Now it is easy to see that non-ambiguous strategies over $\GG(\Aut)$ allow for a one-to-one correspondence with path-based supervisors over $\Aut$, which finally leads to the desired correspondence between \REFproblem{prob:SCT_P} and \REFproblem{prob:OG}.

\begin{proposition}\label{prop:SCTPvsOG}
 Given the premises of \REFthm{thm:SCTPvsOG} the following holds.
 \begin{inparaenum}[(i)]
  \item Let $\check{f}$ be a supervisor solving $(\Aut,\FcS_P,\FcR_S)$ and $\check{h}$ a player $0$ winning strategy over $\GG(\Aut)$ s.t.\
   \begin{equation}\label{equ:ftoh}
 \AllQ{\mu\in q_0(Q^1 Q^0)^*}{\check{h}(\mu)=(\ON{Last}(\mu),\check{f}(\ON{Plays}^{-1}(\mu)))}.
\end{equation}
Then $\check{h}$ is a non-ambiguous winning strategy for $(\GG(\Aut), \FcR_{P\rightarrow S}, \FcS_P)$.
\item Let $\check{h}$ be a non-ambiguous winning strategy for $(\GG(\Aut), \FcR_{P\rightarrow S}, \FcS_P)$ and $\check{f}$ s.t.\ 
 \begin{equation}\label{equ:htof}
  \check{f}(\nu)=\gamma~\text{with}~\gamma\in\set{\ExQ{\mu\in \ON{Plays}(\nu)}{\check{h}(\mu)=(\cdot,\gamma)}}.
 \end{equation}
 Then $\check{f}$ is a path-based supervisor solving $(\Aut,\FcS_P,\FcR_S)$.
 \end{inparaenum}
\end{proposition}

\begin{proof}
First, observe that given $\check{f}$, every $\check{h}$ fulfilling \eqref{equ:ftoh} is non-ambiguous by construction. Conversely, given a non-ambiguous strategy $\check{h}$, \eqref{equ:htof} implies that $\gamma$ is uniquely defined for any $\nu\in x_0X^*$, i.e., $\check{f}$ is a path-based strategy over $\Aut$. Further, given non-ambiguity of $\check{h}$ we can combine the induction from the proof of \REFprop{prop:nonambiguous} and the correspondence used in the proof of \eqref{equ:plays:a} to conclude that 
\begin{equation}\label{equ:helpa}
 \Pomega(\GG,\check{h})=\ON{Plays}(\Pomega(\Aut,\check{f})).
\end{equation}

Now assume \eqref{equ:Pf:a} (equivalently \eqref{equ:Pf:a:new}) holds for $\check{f}$. As the map $\ON{Plays}$ is monotone, this gives $\ON{Plays}(\Pomega(\Aut,\check{f}))\subseteq\ON{Plays}(\Pomega(\Aut,\FcStrong))$. Then it follows from  \eqref{equ:helpa} and \eqref{equ:plays:c} that \eqref{equ:Pf:a} implies that \eqref{equ:fz:strong} holds for $\check{h}$.
Now assume \eqref{equ:Pf:b} holds for $\check{f}$ and observe that the map $\ON{Plays}$ fulfills the following properties: (a) $\LPrefix{\ON{Plays}(A)}=\ON{Plays}(\LPrefix{A})$, and (b) $\ON{Plays}(A\cap B)=\ON{Plays}(A)\cap\ON{Plays}(B)$. With this, it follows from \eqref{equ:helpa} and \eqref{equ:plays:c} that \eqref{equ:Pf:b} implies that \eqref{equ:fz:weak} holds for $\check{h}$.

The reverse direction follows from the same reasoning and is therefore omitted.
\end{proof}

With this, we see that \REFthm{thm:SCTPvsOG} is an immediate corollary of \REFprop{prop:nonambiguous} and \REFprop{prop:SCTPvsOG}.

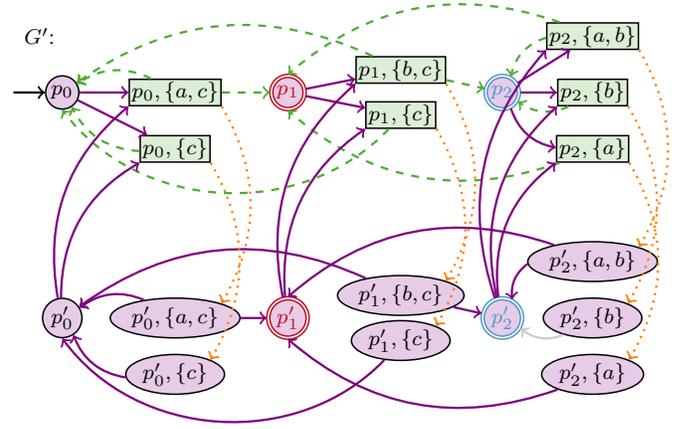
\begin{figure}
\begin{center}
  \begin{tikzpicture}[auto,scale=1.5]
  \begin{footnotesize}
          \node (name) at (-0.2,0.5) {$\GG'$:};
          \node (init) at (-0.5,0) {};
          \node[mystate,fill=violet!20] (p0) at (0,0) {$p_0$};
          \node[mysquare,draw,fill=\lgreen] (p0c) at (1,-0.5) {$p_0, \set{c}$};
          \node[mysquare,draw,fill=\lgreen] (p0ac) at (1,0) {$p_0, \set{a,c}$};
          \node[mystate,accepting,sws-red,fill=violet!20] (p1) at (2,0) {$p_1$};
          \node[mysquare,draw,fill=\lgreen] (p1bc) at (3,0.2) {$p_1, \set{b,c}$};
          \node[mysquare,draw,fill=\lgreen] (p1c) at (3,-0.2) {$p_1, \set{c}$};
          \node[mystate,accepting,sws-blue,fill=violet!20] (p2) at (3.9,0) {$p_2$};
          \node[mysquare,draw,fill=\lgreen] (p2ab) at (4.7,0.5) {$p_2, \set{a,b}$};
          \node[mysquare,draw,fill=\lgreen] (p2b) at (4.7,0) {$p_2, \set{b}$};
          \node[mysquare,draw,fill=\lgreen] (p2a) at (4.7,-0.5) {$p_2, \set{a}$};

          \node[mystate,fill=violet!20] (p0p) at (0,-2) {$p'_0$};
          \node[myestate,draw,fill=violet!20] (p0cp) at (1,-2.5) {$p'_0, \set{c}$};
          \node[myestate,draw,fill=violet!20] (p0acp) at (1,-2) {$p'_0, \set{a,c}$};
          \node[mystate,accepting,sws-red,fill=violet!20] (p1p) at (2,-2) {$p'_1$};
          \node[myestate,draw,fill=violet!20] (p1bcp) at (3,-1.8) {$p'_1, \set{b,c}$};
          \node[myestate,draw,fill=violet!20] (p1cp) at (3,-2.2) {$p'_1, \set{c}$};
          \node[mystate,accepting,sws-blue,fill=violet!20] (p2p) at (3.9,-2) {$p'_2$};
          \node[myestate,draw,fill=violet!20] (p2abp) at (4.7,-1.5) {$p'_2, \set{a,b}$};
          \node[myestate,draw,fill=violet!20] (p2bp) at (4.7,-2) {$p'_2, \set{b}$};
          \node[myestate,draw,fill=violet!20] (p2ap) at (4.7,-2.5) {$p'_2, \set{a}$};

\SFSAutomatEdge{init}{}{p0}{}{}  
\SFSAutomatEdge{p0}{}{p0c}{violet}{}
\SFSAutomatEdge{p0}{}{p0ac}{violet}{}
\SFSAutomatEdge{p0ac}{}{p0}{bend right,sws-green, dashed}{}
\SFSAutomatEdge{p0ac}{}{p1}{sws-green, dashed}{}
\SFSAutomatEdge{p0c}{}{p0}{bend left,sws-green, dashed}{}
\SFSAutomatEdge{p1}{}{p1bc}{violet}{}
\SFSAutomatEdge{p1}{}{p1c}{violet}{}
\SFSAutomatEdge{p1bc}{}{p2}{sws-green, dashed}{}
\SFSAutomatEdge{p1bc}{}{p0}{bend right,sws-green, dashed}{}
\SFSAutomatEdge{p1c}{}{p0.south}{bend left=50,sws-green, dashed}{}
\SFSAutomatEdge{p2}{}{p2ab}{violet}{}
\SFSAutomatEdge{p2}{}{p2a}{bend right,violet}{}
\SFSAutomatEdge{p2}{}{p2b}{violet}{}
\SFSAutomatEdge{p2ab}{}{p2}{bend right,sws-green, dashed}{}
\SFSAutomatEdge{p2b}{}{p2}{bend left,sws-green, dashed}{}
\SFSAutomatEdge{p2ab}{}{p1.north}{bend right,sws-green, dashed}{}
\SFSAutomatEdge{p2a}{}{p1.south}{bend left,sws-green, dashed}{}

\SFSAutomatEdge{p0acp}{}{p0p}{bend right,violet}{}
\SFSAutomatEdge{p0acp}{}{p1p}{violet}{}
\SFSAutomatEdge{p0cp}{}{p0p}{bend left,violet}{}
\SFSAutomatEdge{p1bcp}{}{p2p}{violet}{}
\SFSAutomatEdge{p1bcp}{}{p0p}{bend right,violet}{}
\SFSAutomatEdge{p1cp}{}{p0p.south}{bend left=50,violet}{}
\SFSAutomatEdge{p2abp}{}{p2p}{bend right,violet}{}
\SFSAutomatEdge{p2bp}{}{p2p}{bend left,\gr}{}
\SFSAutomatEdge{p2abp}{}{p1p.north}{bend right,violet}{}
\SFSAutomatEdge{p2ap}{}{p1p.south}{bend left,violet}{}

%

\SFSAutomatEdge{p0ac.south east}{}{p0acp.north east}{bend left, orange, dotted}{}
\SFSAutomatEdge{p0c.south east}{}{p0cp.north east}{bend left, orange, dotted}{}
\SFSAutomatEdge{p1bc.south east}{}{p1bcp.north east}{bend left, orange, dotted}{}
\SFSAutomatEdge{p1c.south east}{}{p1cp.north east}{bend left, orange, dotted}{}
\SFSAutomatEdge{p2ab.south east}{}{p2abp.north east}{bend left, orange, dotted}{}
\SFSAutomatEdge{p2b.south east}{}{p2bp.north east}{bend left, orange, dotted}{}
\SFSAutomatEdge{p2a.south east}{}{p2ap.north east}{bend left, orange, dotted}{}

%

\SFSAutomatEdge{p0p}{}{p0ac.south west}{violet, bend left}{}
\SFSAutomatEdge{p0p}{}{p0c.south west}{violet, bend left}{}
\SFSAutomatEdge{p1p}{}{p1bc.south west}{violet, bend left}{}
\SFSAutomatEdge{p1p}{}{p1c.south west}{violet, bend left}{}
\SFSAutomatEdge{p2p}{}{p2ab.south west}{violet, bend left}{}
\SFSAutomatEdge{p2p}{}{p2b.south west}{violet, bend left}{}
\SFSAutomatEdge{p2p}{}{p2a.south west}{violet, bend left}{}

  \end{footnotesize}
\end{tikzpicture}
\end{center}
 \vspace{-0.8cm}
 \caption{Obliging game graph expansion of $\GG$ in \REFfig{fig:MtoG} as discussed in \REFsec{sec:example:strategy} (see \cite{chatterjee2010obliging} for a formalization). The plant can decide to chose the next event by herself (dashed green transitions) or to let the controller decide on her behalf (dotted orange transition followed by a solid violet one).}\label{fig:OGG}
 \vspace{-0.4cm}
\end{figure}

\subsection{Example}\label{sec:example:strategy}
The technical reduction from obliging games to games with $\omega$-regular winning conditions (see \REFthm{thm:obliging-games}) can be found in \cite{chatterjee2010obliging}.
We give an intuitive explanation of this construction by applying it to our example and thereby constructing a winning strategy for the obliging game $(\GG(\Aut), \FcR_{P\rightarrow S}, \FcS_P)$ over the game graph $\GG(\Aut)$ depicted in \REFfig{fig:MtoG}.

As the first step of this construction, we double the state space of $\GG$ resulting in an upper and a lower part (see \REFfig{fig:OGG}). 
The upper part is a copy of the old state space while in the lower part all states become control player states 
(indicated by their violet ellipse shape). 
Now we run the following  Gedankenexperiment: 
in every (rectangular green) state the plant can choose between deciding on the next executed event 
by herself or allowing the controller to make this choice for her. 
In the first case the play stays within the upper part (using a dashed green transition), 
while in the second case the play moves to the lower part (using a dotted orange transition) and 
the controller decides the next move on behalf of the plant (by taking an available solid violet transition). 
In each case, the play moves to a control player's state ($p_i$ (top) or $p_i'$ (bottom), with $i\in\set{0,1,2}$). 
In both cases, the controller chooses a control pattern $\gamma$ and by this always moves 
to the rectangular green state $(p_i,\gamma)$ in the upper part. 
Here, it is again the choice of the plant to either stay in the original (top) game or to move to the bottom copy. 

With this modified game in mind, we can interpret the two copies of the game graph as follows. 
In the top one, the controller is only concerned with fulfilling the specification, i.e., 
solving a standard two-player game with the winning condition $\FcR_{P\rightarrow S}$ in \eqref{ex:rabin}.

The bottom copy of the game makes sure that the resulting strategy is non-conflicting. 
Within the outlined Gedankenexperiment, this is ensured by the fact that at any point in time, 
the plant can decide to hand over all future choices of the next events to the controller and the controller must be able to 
\emph{demonstrate} that the liveness condition of the plant (i.e., $\FcS_P$) remains satisfiable along with satisfying $\FcR_{P\rightarrow S}$. 
Hence, from every reachable state in the top game, the controller must be able to give \emph{one explicit trace} 
which visits both $p_1'$ and $p_2'$ always eventually again. 
This prevents the controller from moving to a state in the top game where the plant's assumptions are persistently violated.
It should be noted that the synthesis problem over the lower game graph is actually much simpler, as it 
only involves one player (namely the controller) and thereby reduces to a simple path search. 

A gracious strategy in the original obliging game is extracted from this Gedankenexperiment as follows. First, we consider the upper and the lower game in \REFfig{fig:OGG} separately. For the upper game, we know that a supervisor disabling events $a$ and $b$ in every state is winning w.r.t.\ $\FcR_{P\rightarrow S}$ (see \REFex{expl:rabin}). Call this strategy $h^{\uparrow}$.
We can assume w.l.o.g.\ that  $h^{\uparrow}$ is memoryless\footnote{A strategy $\fz{}: q_{init}(\Qo\Qz)^*\rightarrow \Qo$ is memory-less if for all $\nu,\nu'\in (\Qz\Qo)^*$ and $q\in\Qz$ holds that $\fz{\nu q}=\fz{\nu' q}$. I.e., the strategy bases its choice of control patterns purely on the current state of the play.}, because, in any Rabin game, if there is a winning strategy, there is also a memoryless one.
This strategy forces the plant to always remain in $p_0$ and wins in the upper game by vacuously satisfying the implication. 
  
For the lower part, consider the blue transitions in \REFfig{fig:OGG_trace}, indicating an infinite trace from every state visiting both $p_1$ and $p_2$ infinitely often, fulfilling both $\FcStrong=\FcR_{P\rightarrow S}$ and $\FcWeak=\FcS_P$. 
This path immediately defines a memoryless plant and a control player strategy which we denote $g^{\downarrow}$ and $h^{\downarrow}$, respectively.

\begin{figure}
\begin{center}
   \begin{tikzpicture}[auto,scale=1.5]
  \begin{footnotesize}
          \node[mystate,fill=violet!20] (p0p) at (0,-2) {$p'_0$};
          \node[myestate,draw,fill=violet!20] (p0cp) at (1,-2.5) {$p'_0, \set{c}$};
          \node[myestate,draw,fill=violet!20] (p0acp) at (1,-2) {$p'_0, \set{a,c}$};
          \node[mystate,accepting,sws-red,fill=violet!20] (p1p) at (2,-2) {$p'_1$};
          \node[myestate,draw,fill=violet!20] (p1bcp) at (3,-1.8) {$p'_1, \set{b,c}$};
          \node[myestate,draw,fill=violet!20] (p1cp) at (3,-2.2) {$p'_1, \set{c}$};
          \node[mystate,accepting,sws-blue,fill=violet!20] (p2p) at (4,-2) {$p'_2$};
          \node[myestate,draw,fill=violet!20] (p2abp) at (5,-1.5) {$p'_2, \set{a,b}$};
          \node[myestate,draw,fill=violet!20] (p2bp) at (5,-2) {$p'_2, \set{b}$};
          \node[myestate,draw,fill=violet!20] (p2ap) at (5,-2.5) {$p'_2, \set{a}$};

\SFSAutomatEdge{p0p}{}{p0cp}{\gr}{}
\SFSAutomatEdge{p0p}{}{p0acp}{\gr}{}
\SFSAutomatEdge{p0acp}{}{p0p}{bend right,\gr}{}
\SFSAutomatEdge{p0acp}{}{p1p}{\gr}{}
\SFSAutomatEdge{p0cp}{}{p0p}{bend left,\gr}{}
\SFSAutomatEdge{p1p}{}{p1bcp}{\gr}{}
\SFSAutomatEdge{p1p}{}{p1cp}{\gr}{}
\SFSAutomatEdge{p1bcp}{}{p2p}{\gr}{}
\SFSAutomatEdge{p1bcp}{}{p0p}{bend right,\gr}{}
\SFSAutomatEdge{p1cp}{}{p0p.south}{bend left=50,\gr}{}
\SFSAutomatEdge{p2p}{}{p2abp}{\gr}{}
\SFSAutomatEdge{p2p}{}{p2ap}{bend right,\gr}{}
\SFSAutomatEdge{p2p}{}{p2bp}{\gr}{}
\SFSAutomatEdge{p2abp}{}{p2p}{bend right,\gr}{}
\SFSAutomatEdge{p2bp}{}{p2p}{bend left,\gr}{}
\SFSAutomatEdge{p2abp}{}{p1p.north}{bend right,\gr}{}
\SFSAutomatEdge{p2ap}{}{p1p.south}{bend left,\gr}{}

\SFSAutomatEdge{p0p}{}{p0acp}{blue,dashed}{}
\SFSAutomatEdge{p0acp}{}{p1p}{blue,dashed}{}
\SFSAutomatEdge{p0cp}{}{p0p}{bend left,blue,dashed}{}
\SFSAutomatEdge{p1p}{}{p1bcp}{blue}{}
\SFSAutomatEdge{p1bcp}{}{p2p}{blue}{}
\SFSAutomatEdge{p1cp}{}{p0p.south}{bend left=50,blue,dashed}{}
\SFSAutomatEdge{p2p}{}{p2ap}{bend right,blue}{}
\SFSAutomatEdge{p2bp}{}{p2p}{bend left,blue,dashed}{}
\SFSAutomatEdge{p2abp}{}{p1p.north}{bend right,blue,dashed}{}
\SFSAutomatEdge{p2ap}{}{p1p.south}{bend left,blue}{}

%
%
%

  \end{footnotesize}
         \end{tikzpicture} 
\end{center}
 \vspace{-0.8cm}
 \caption{Witness of a path from every reachable state of $\GG(\Aut)$ (dashed) to a loop (solid) satisfying 
plant and specification makings always again. 
This defines a plant and control player strategy denoted by $g^{\downarrow}$ and $h^{\downarrow}$, respectively.}\label{fig:OGG_trace}
 \vspace{-0.4cm}
\end{figure}
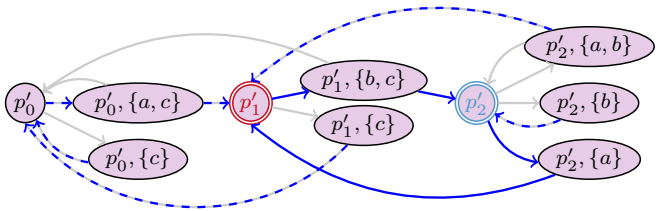

Given $h^{\uparrow}$, $g^{\downarrow}$, and $h^{\downarrow}$, we can combine them into a solution to the original synthesis problem over 
$\GG(\Aut)$ (and therefore $\Aut$) in \REFfig{fig:MtoG}, 
by adding one extra bit of memory to the controller. That is, the resulting strategy will base its decision on the current state and an additional binary-valued variable $m$ which tracks, whether the system executes a move contained in $g^{\downarrow}$ ($m=1$) or not ($m=0$).
If $m=1$, the controller executes the unique pattern chosen by $h^{\downarrow}$ in the next state. 
Otherwise, it operates according to $h^{\uparrow}$.

For the particular choices of strategies in this example we see that the only allowed event $c$ 
in $(p_0,\gamma)$ is part of $g^{\downarrow}$ and therefore triggers $h^{\downarrow}$. Hence, the actual closed loop allows the plant to 
move to $p_1$ next. 
If it does so, $h^{\downarrow}$ remains active as this move is again contained in $g^{\downarrow}$ (see \REFfig{fig:OGG_trace}). If the plant decides to stay in $p_0$, $h^{\uparrow}$ becomes active again.
Intuitively, the controller tracks whether the plant is trying to make progress towards fulfilling her liveness condition. 
If so, he is cooperating with her to achieve this goal.

\subsection{Algorithm}

The reduction outlined in the previous section via \REFthm{thm:SCTvsSCTP} and \REFthm{thm:SCTPvsOG} enables us to solve a given supervisory synthesis problem (\REFproblem{prob:SCT}) over 
a plant model $(\Lstar_P,\Lomega_P)$ w.r.t.\ a specification $\Lomega_S$ and a set of uncontrollable events 
$\Sigma_{uc}\subseteq\Sigma$ through the following steps:
\begin{compactenum}
 \item Construct a Street/Rabin synthesis automaton $(\Aut,\FcS_P,\FcR_S)$ as in \REFprop{prop:StreetRabinAut}. 
 \item Extend $\Aut$ into a game graph $\GG(\Aut)$ as in \REFdef{def:MtoG}.
 \item Solve the obliging game $(\GG(\Aut),\FcR_{P\rightarrow S},\FcS_P)$ via its reduction to standard $\omega$-regular games (see \REFthm{thm:obliging-games}).
 \item If the obliging game has no solution, also \REFproblem{prob:SCT} has no solution (see \REFthm{thm:SCTvsSCTP} and \REFthm{thm:SCTPvsOG}).
 \item If the obliging game allows for a control strategy $h$, compute its induced non-ambiguous strategy $\check{h}$ as in \eqref{equ:checkh}.
 \item Reduce $\check{h}$ to a path-based supervisor via \eqref{equ:ftoh}, which in turn defines the event-based supervisor $f$ via \eqref{equ:fvscheckf}.
 \item Then $f$ solves \REFproblem{prob:SCT} (see \REFthm{thm:SCTvsSCTP} and \REFthm{thm:SCTPvsOG}).
\end{compactenum}

\smallskip
The complexity of this algorithm can be derived from \REFthm{thm:obliging-games} in the following way. Given a finite state machine $\Aut$ with $n$ states we get a game graph $\GG(\Aut)$ with $n2^{|\ASigmac|}$ states. Further, given the Streett and Rabin conditions $\FcS_P$ and $\FcR_S$ with $l$ and $k$ pairs, we get an obliging game having a strong Rabin condition with $l+k$ pairs and a weak Streett condition with $l$ pairs. Finally, a  parity game with $\tilde{n}$ states and $\tilde{k}$ colors can be solved in $O({\tilde{n}}^{\tilde{k}})$ time. Hence, our solution can be computed in time $O((n2^{|\ASigmac|}(l+k)^2(l+k)!2^{O(l)})^{2(l+k)+2})$. If there is a supervisor, then there is a supervisor using $2(l+k)\cdot2^{O(l)}$ memory.

It should further be noted that checking if there is a path-based supervisor from a state is NP-complete \cite{thistle1995};
this already holds for a trivial liveness assumption  for the plant (i.e., $\Lomega_P=\Sigma^\omega$) as solving Rabin games is NP-complete \cite{EmersonJutla88}.
While our algorithm is sound and complete, it is possible that there is a more direct symbolic algorithm on the state space of the two-person
game that yields a more efficient implementation. Such an algorithm is given in \cite{majumdar2019environmentally} for the special case where $\Fc_P$ and $\Fc_S$ are each a generalized Büchi winning condition. 
We postpone the generalization of this algorithm to future work.
.

\bibliographystyle{IEEEtran}
\bibliography{tmuni,AKSbib}

 \end{document}